\documentclass[bib]{ba}

\usepackage{graphicx,hyperref, natbib}

\usepackage{bayes}

\RequirePackage[OT1]{fontenc}
\RequirePackage{amsthm,amsmath}

\usepackage{color}

\usepackage{array}
\usepackage{dsfont}
\usepackage{graphicx}
\usepackage{mathrsfs}
\usepackage{amssymb}
\usepackage{amsbsy}
\usepackage{amsfonts}
\usepackage{latexsym}

\numberwithin{equation}{section}
\theoremstyle{plain}
\newtheorem{thm}{Theorem}[section]
\newtheorem{prop}{Proposition}[section]
\newtheorem{rmk}{Remark}[section]
\newtheorem{cor}{Corollary}[section]
\newtheorem{lem}{Lemma}[section]

\def\1{\mathbf{1}}

\newcommand{\E}{\ensuremath{\mathbb{E}}}

\renewcommand{\P}{\ensuremath{\mathbb{P}}}

\newcommand{\R}{\ensuremath{\mathbb{R}}}

\renewcommand{\L}[1]{\ensuremath{\mathbb{L}^{#1}}}

\renewcommand{\L}{\ensuremath{\mathbb{L}}}

\renewcommand{\d}{\mathrm{d}}

\newcommand{\pg}[1]{\left\{#1\right\}}
\newcommand{\pq}[1]{\left[#1\right]}

\newcommand{\vir}[1]{``#1''}

\usepackage{amsmath}

\makeatletter
\newcommand{\pushright}[1]{\ifmeasuring@#1\else\omit\hfill$\displaystyle#1$\fi\ignorespaces}
\newcommand{\pushleft}[1]{\ifmeasuring@#1\else\omit$\displaystyle#1$\hfill\fi\ignorespaces}
\makeatother

\begin{document}

\allowdisplaybreaks

% AUTHOR:  Change "xart" to your insert's file name.

\inserttype[]{article}
\renewcommand{\thefootnote}{\fnsymbol{footnote}}

\author{S. Donnet and V. Rivoirard and J. Rousseau and C. Scricciolo}{
 \fnms{Sophie}
 \snm{Donnet}
 \footnotemark[1]\ead{email1@example.com}
and
\fnms{Vincent}
\snm{Rivoirard}
\footnotemark[2]\ead{email2@example.com}
and
\fnms{Judith}
\snm{Rousseau}
\footnotemark[3]\ead{email3@example.com}
and
\fnms{Catia}
\snm{Scricciolo}
\footnotemark[4]\ead{email4@example.com}
}

\title[Posterior concentration rates for Aalen counting processes]
{Posterior concentration rates for counting processes with Aalen multiplicative intensities}

\maketitle

\footnotetext[1]{
 MIA. INRA. UMR0518, AgroParisTech,
 \href{donnet@cimat.mx}{donnet@cimat.mx}
}

\footnotetext[2]{
 CEREMADE, Universit\'e Paris Dauphine,
 \href{rivoirard@ceremade.dauphine.fr}{rivoirard@ceremade.dauphine.fr}}

 \footnotetext[3]{
 ENSAE-CREST,
 \href{rousseau@ceremade.dauphine.fr}{rousseau@ceremade.dauphine.fr}}

 \footnotetext[4]{
 Department of Decision Sciences, Bocconi University,
 \href{catia.scricciolo@unibocconi.it}{catia.scricciolo@unibocconi.it}
 }

\renewcommand{\thefootnote}{\arabic{footnote}}

\begin{abstract}
We provide general conditions to derive posterior concentration rates for Aalen counting processes.
The conditions are designed to resemble %be as close as possible to
those proposed in the literature for the problem of density estimation, for instance in \citet{ghosal:ghosh:vdv:00}, so that existing results on density estimation
can be adapted to the present setting.
We apply the general theorem to some prior models including Dirichlet process mixtures of uniform densities
to estimate monotone non-increasing intensities and log-splines. % and random histograms.
%A simulation study for inhomogeneous Poisson processes  and censored survival models also illustrates our results.

\end{abstract}

%%%%%%%%%%%%%%%%%%%%%%%%%%%%%%%%%%%%%%%%%%%%%%%%%%%%%%%%%%%%%%%%%%%%%%%%%%%%%%%%%%%%%%%%%%%%%%%%%%

\section{Introduction} \label{intro}
Estimation of the intensity function of a point process is an important statistical problem with a long
history. Most methods were initially employed for estimating intensities
assumed to be of parametric or nonparametric form in Poisson point processes.
However, in many fields such as genetics, seismology and neuroscience, the probability of observing a
new occurrence of the studied temporal process may depend on covariates and, in this case,
the intensity of the process is random so that such a feature is not captured by a
classical Poisson model. Aalen models constitute a natural extension of Poisson models
that allow taking into account this aspect. \cite{aalen1978} revolutionized point processes
analysis developing a unified theory for frequentist nonparametric inference
of multiplicative intensity models which, besides the Poisson model and
other classical models such as right-censoring and %the model associated with
Markov processes with finite state space, described in Section~\ref{subsec:notations},
encompass birth and death processes as well as branching processes.
We refer the reader to  \citet{ABGK} for a presentation of Aalen processes including various other illustrative examples.
Classical probabilistic and statistical results about Aalen processes can be found in  \citet{Karr}, \citet{ABGK}, \citet{MR1950431,MR2371524}. Recent nonparametric frequentist methodologies based on penalized least-squares contrasts have been proposed by \citet{MR2235573, MR2440441}, \citet{MR2884230} and \citet{Aalen-Patricia}. In the high-dimensional setting, more specific results have been established by \citet{GG} and \citet{HRR} who consider Lasso-type procedures.

Bayesian nonparametric inference for inhomogeneous Poisson point processes has been
considered by \citet{lo1982} who develops a prior-to-posterior analysis
for weighted gamma process priors to model intensity functions. In the same spirit,
\citet{kuoandghosh} employ several classes of nonparametric priors,
including the gamma, the beta and the extended gamma processes. Extension to
multiplicative counting processes has been treated in
\citet{loandweng1989}, who model intensities as kernel mixtures with mixing measure distributed according to a weighted
gamma measure on the real line. Along the same lines, \citet{ishwaranandjames} develop computational procedures for Bayesian
non- and semi-parametric multiplicative intensity models using kernel mixtures of weighted gamma measures.
Other papers have mainly focussed on exploring prior distributions on intensity functions
with the aim of showing that Bayesian nonparametric inference for inhomogeneous Poisson
processes can give satisfactory results in applications, see, \emph{e.g.},
\citet{kottas:sansò:2007}.

Surprisingly, leaving aside the recent work of \citet{belitser:serra:vanzanten},
which deals with optimal convergence rates for estimating intensities
in inhomogeneous Poisson processes,
there are no results in the literature concerning aspects of the frequentist asymptotic behaviour of posterior distributions,
like consistency and rates of convergence, for intensity estimation of general Aalen models.
In this paper, we extend their results to general Aalen multiplicative intensity models.
Quoting \citet{loandweng1989}, \vir{the idea of our approach is that estimating a density and estimating a hazard rate are analogous affairs,
and a successful attempt of one generally leads to a feasible approach for the other}.
Thus, in deriving general sufficient conditions for assessing posterior contraction rates
in Theorem \ref{th:gene:aalen} of Section \ref{aalen},
we attempt at giving conditions which resemble those proposed by \citet{ghosal:ghosh:vdv:00}
for density estimation with independent and identically distributed (i.i.d.) observations.
This allows us to then derive in Section \ref{sec:priors} posterior contraction rates for different
families of prior distributions, such as Dirichlet mixtures of uniform densities to estimate monotone
non-increasing intensities and log-splines, by an adaptation of existing results on density estimation.
Detailed proofs of the main results are reported in Section \ref{app:proof}.
Auxiliary results concerning the control of the Kullback-Leibler divergence for intensities in Aalen models
and existence of tests, which, to the best of our knowledge, are derived here for the first time and can also
be of independent interest, are presented in Section \ref{app:Aalen} and in Section \ref{app}.

%%%%%%%%%%%%%%%%%%%%%%%%%%%%
\subsection{Notation and set-up} \label{subsec:notations}
We observe a counting process $N$ and denote by $(\mathcal G_t)_t$ its adapted  filtration. Let  $\Lambda$ be the compensator of $N$.
We assume it satisfies the condition $\Lambda_t<\infty$ almost surely for every $t$. Recall that  $(N_t-\Lambda_t)_t$ is a zero-mean $(\mathcal G_t)_t$-martingale.
We assume that $N$ obeys the {\it Aalen multiplicative intensity model}
$$\d\Lambda_t=Y_t\lambda(t)\d t,$$
where $\lambda$ is a non-negative deterministic function called \emph{intensity function} in the sequel and $(Y_t)_t$ is a non-negative predictable process. Informally,
\begin{equation}\label{informally}
\P[N[t,\,t+\d t]\geq 1\mid \mathcal G_{t^-}]=Y_t\lambda(t)\d t,
\end{equation}
see \cite{ABGK}, Chapter III. In this paper, we are interested in asymptotic results: both $N$ and $Y$ depend on an integer $n$ and we study estimation of $\lambda$ (not depending on $n$) when $T$ is kept fixed and $n\rightarrow\infty$.
The following special cases motivate the interest in this model.
\vspace*{-0.4cm}\subsubsection*{Inhomogeneous Poisson processes}
\vspace*{-0.2cm}We observe $n$ independent Poisson processes with common intensity $\lambda$. This model is equivalent to the model where we observe a Poisson process with intensity $n\times\lambda$, so it corresponds to the case $Y_t\equiv n$.
\vspace*{-0.4cm}\subsubsection*{Survival analysis with right-censoring}
\vspace*{-0.2cm}This model is popular in biomedical problems. We have $n$ patients and, for each patient $i$, we observe  $(Z_i,\,\delta_i)$, with $Z_i=\min\{X_i,\,C_i\}$, where $X_i$ represents the lifetime of the patient, $C_i$ is the independent censoring time and $\delta_i=\1_{X_i\leq C_i}$. In this case, we set $N_t^i=\delta_i\times \1_{Z_i\leq t}$, $Y_t^i=\1_{Z_i\geq t}$ and $\lambda$ is the hazard rate of the $X_i$'s: if $f$ is the density of $X_1$, then
$\lambda(t)=f(t)/\P(X_1\geq t).$
Thus, $N$ (respectively $Y$) is obtained by aggregating the $n$ independent processes $N^i$'s (respectively the $Y^i$'s): for any $t\in [0,\,T]$,
$N_t=\sum_{i=1}^n N^i_t$ and $  Y_t=\sum_{i=1}^n Y^i_t$.
\vspace*{-0.4cm}\subsubsection*{Finite state Markov processes}
\vspace*{-0.2cm}Let  $X=(X(t))_t$ be a Markov process with finite state space $\mathbb{S}$ and right-continuous sample paths. We assume the existence of  integrable {\it transition intensities} $\lambda_{hj}$ from state $h$ to state $j$ for $h\not= j$. We assume we are given $n$ independent copies of the process $X$, denoted by $X^1,\,\ldots,\,X^n$. For any $i\in\{1,\,\ldots,\,n\}$,
let $N_t^{ihj}$ be the number of direct transitions for $X^i$ from $h$ to $j$  in $[0,\,t]$, for $h\not=j$. Then, the intensity of the multivariate counting process $\frak{N}^i=(N^{ihj})_{h\not=j}$ is $(\lambda_{hj}Y^{ih})_{h\not=j},$ with $Y^{ih}_t=\1_{\{X^i(t^-)=h\}}$. As before, we can consider $\frak{N}$ (respectively $Y^h$) by aggregating the processes $\frak{N}^i$ (respectively the $Y^{ih}$'s): $\frak{N}_t =\sum_{i=1}^n \frak{N}^{i}_t$, $Y_t^h=\sum_{i=1}^n Y^{ih}_t$ and $t\in [0,\,T]$.
The intensity of each component $(N_t^{hj})_t$ of $(\frak{N}_t)_t$  is then $(\lambda_{hj}(t)Y_t^h)_t$.
We refer the reader to \cite{ABGK}, p.~126, for more details. In this case, $N$ is either one of the $N^{hj}$'s or the aggregation of some processes for which the $\lambda_{hj}$'s are equal.

We now state some conditions concerning the asymptotic behavior of $Y_t$ under the true intensity function $\lambda_0$.
Define $\mu_n(t):=\E_{\lambda_0}^{(n)}\left[ Y_t\right]$ and $\tilde\mu_n(t):=n^{-1}\mu_n(t)$.
We assume the existence of a non-random
set $\Omega \subseteq [0,\, T]$ such that there are constants $m_1$ and $m_2$ satisfying
\begin{equation}\label{ass:Y1}
m_1\leq \inf_{t\in\Omega}\tilde \mu_n(t)\leq \sup_{t\in\Omega}\tilde \mu_n(t)\leq m_2\quad\mbox{for every $n$ large enough,}
\end{equation}
and there exists $\alpha\in (0,\,1)$ such that, if
$\Gamma_n:=\{ \sup_{t\in\Omega}|n^{-1}Y_t-\tilde\mu_n(t)|\leq \alpha m_1\}\cap \{\sup_{t\in\Omega^c}  Y_t = 0\}$,
where $\Omega^c$ is the complement of $\Omega$ in $[0,\,T]$, then
\begin{equation}\label{ass:Y2}
\lim_{n\rightarrow\infty} \P_{\lambda_0}^{(n)}\left(\Gamma_n\right)= 1.
\end{equation}
%\begin{remark}
%Since our results are asymptotic in nature, we can assume, without loss of generality, that \eqref{ass:Y1} is true only for $n$ large enough.
%\end{remark}
We only consider estimation over $\Omega$ ($N$ is almost surely empty on $\Omega^c$) and define the parameter space as
 $\mathcal F = \{ \lambda : \Omega \rightarrow \R_+\,\, |\,  \int_{\Omega} \lambda(t)\d t < \infty \}$.
Let %the true intensity function
$\lambda_0\in\mathcal F$.

For inhomogeneous Poisson processes, conditions \eqref{ass:Y1} and \eqref{ass:Y2} are trivially satisfied for $\Omega=[0,\,T]$ since $Y_t\equiv\mu_n(t)\equiv n.$ For right-censoring models, with $Y_t^i=\1_{Z_i\geq t}$, $i=1,\,\ldots,\,n$, we denote by $\Omega$ the support of the $Z_i$'s and by $M_{\Omega}=\max\Omega\in \overline\R_+.$ Then, \eqref{ass:Y1} and \eqref{ass:Y2} are satisfied if $M_{\Omega}>T$ or $M_{\Omega}\leq T$ and $\P(Z_1=M_{\Omega})>0$ (the concentration inequality is implied by an application of the DKW inequality).

We denote by $\| \cdot \|_1$ the $\L_1$-norm over $\mathcal{F}:$ for all $\lambda,\,\lambda'\in \mathcal F$,
$\| \lambda - \lambda' \|_1  = \int_{\Omega} | \lambda(t)  -  \lambda'(t)  | \d t$.

%%%%%%%%%%%%%%%%%%%%%%%%%%%%%%%%%%%%%%%%%%%%%%%%%%%%%%%%%%%%%%%%%%%%%%%%%%%%%%%%%%%%%%%%%%%%%%

%%%%%%%%%%%%%%%%%%%
\section{Posterior contraction rates for Aalen counting processes} \label{aalen}

In this section, we present the main result providing general sufficient conditions for assessing concentration rates of
posterior distributions of intensities in general Aalen models. Before stating the theorem, we need to introduce some
more notation.

For any $\lambda \in \mathcal F$,
we introduce the following parametrization $\lambda  = M_\lambda\times\bar \lambda$, where $M_\lambda= \int_{\Omega} \lambda(t)\d t$ and $\bar \lambda  \in \mathcal F_1$, with
$\mathcal F_1 = \{ \lambda \in \mathcal F: \ \int_\Omega \lambda(t) \d t = 1\}.$
For the sake of simplicity, in this paper we restrict attention to
the case where $M_\lambda$ and $\bar\lambda$ are \emph{a priori} independent so that the
prior probability measure $\pi$ on $\mathcal F$ is the product measure $\pi_1 \otimes \pi_M$, where $\pi_1$ is a probability measure on $\mathcal F_1$ and $\pi_M $ is a probability measure on $\R_+$.
Let $v_n$ be a positive sequence such that $v_n\rightarrow0$ and $n v_n^2\rightarrow\infty$.  For every $j \in\mathbb{N}$, we define
$$\bar S_{n,j} = \left\{ \bar \lambda \in \mathcal F_1:\,\| \bar \lambda - \bar \lambda_0 \|_1 \leq 2(j+1)v_n/M_{\lambda_0}\right\},$$
where
$M_{\lambda_0}=\int_{\Omega}\lambda_0(t)\d t$ and $\bar\lambda_0=M_{\lambda_0}^{-1}\lambda_0$.
For $H>0$ and $k\geq 2$,  if $k_{[2]} = \min \{ 2^\ell:\, \ell\in \mathbb N,\, 2^\ell \geq k\}$, we define
%\begin{small}
\begin{equation*}
\begin{split}
\bar B_{k,n}(\bar \lambda_0;\,v_n,\,H)=& \bigg\{ \bar \lambda \in \mathcal F_1:\, \ h^2(\bar \lambda_0,\, \bar \lambda)
\leq v_n^2/(1 + \log  \| \bar \lambda_0/\bar \lambda\|_\infty ),\\& \qquad\qquad \max_{2\leq j \leq k_{[2]} }E_j(\bar \lambda_0 ;\, \bar \lambda) \leq v_n^2, \,\,
%\right.\\&\left.  \hspace{2cm}\
\|\bar \lambda_0/\bar \lambda\|_\infty \leq n^H,\,\, \left\|\bar \lambda\right\|_\infty \leq H \bigg\},
\end{split}
\end{equation*}
%\end{small}
where $h^2(\bar \lambda_0,\, \bar \lambda)=\int_\Omega(\sqrt{\bar \lambda_0(t)}-\sqrt{\bar \lambda(t)})^2\d t$ is the squared Hellinger distance between $\bar \lambda_0$ and $\bar \lambda$,
$\|\cdot\|_\infty$ stands for the sup-norm and $E_j(\bar \lambda_0;\, \bar \lambda):= \int_{\Omega} \bar \lambda_0(t) | \log \bar \lambda_0  (t) - \log \bar \lambda (t)  |^j\d t$. In what follows, for any set $\Theta$ equipped with a semi-metric $d$ and any real number $\epsilon>0$,
we denote by $D( \epsilon,\, \Theta ,\, d)$ the $\epsilon$-packing number of $\Theta$, that is, the maximal number of points in $\Theta$
such that the $d$-distance between every pair is at least $\epsilon$. Since $D( \epsilon,\, \Theta ,\, d)$ is bounded above by the $(\epsilon/2)$-covering number, namely, the minimal number of balls of $d$-radius $\epsilon/2$ needed to cover $\Theta$, with abuse of language,
we will just speak of covering numbers. We denote by $\pi(\cdot\mid N)$ the posterior distribution of the intensity function $\lambda$,
given the observations of the process $N$.

\begin{thm}\label{th:gene:aalen}
Assume that conditions \eqref{ass:Y1} and \eqref{ass:Y2} are satisfied and that,
for some $k\geq 2$, there exists a constant $C_{1k}>0$ such that
\begin{equation}\label{moment}
\E_{\lambda_0}^{(n)}\left[ \left(  \int_\Omega [ Y_t-\mu_n(t)]^2\d t \right)^k  \right] \leq C_{1k} n^k.
\end{equation}
Assume that the prior $\pi_M$ on the mass $M$ is absolutely continuous with respect to Lebesgue measure and has
positive and continuous density on $\R_+$, while
the prior $\pi_1$ on $\bar \lambda$ satisfies the following conditions for some constant $H>0$:
\begin{itemize}
\item[$(i)$] there exists $\mathcal F_n \subseteq \mathcal F_1$ such that, for a positive sequence
$v_n=o(1)$ and $v_n^2 \geq (n/\log n)^{-1}$,
$$\pi_1\left( \mathcal F_n^c \right) \leq e^{ -(\kappa_0+2) n v_n^2}\pi_1(\bar B_{k,n}(\bar \lambda_0;\,v_n,\,H)),  $$
with
\begin{equation}\label{kap0}
\kappa_0=m_2^2M_{\lambda_0}
\pg{\frac{4}{m_1}\pq{1+\log\left(\frac{m_2}{m_1}\right)}
\left(1+\frac{m_2^2}{m_1^2}\right)+\frac{m_2(2M_{\lambda_0}+1)^2}{m_1^2M_{\lambda_0}^2}},
\end{equation}
 and, for any $\xi,\,\delta>0$,
$$\log D( \xi,\,\mathcal F_n,\,\| \cdot \|_1 )\leq n \delta \quad\mbox{ for all $n$ large enough;}$$

\item[$(ii)$] for all $\zeta,\,\delta>0$,
%,\,\beta>0$,
there exists $J_0>0$ such that, for every $j\geq J_0$, %$j\in[J_0,\,\beta/v_n]$,
$$\frac{ \pi_1(\bar S_{n,j}) }{ \pi_1(\bar B_{k,n}(\bar \lambda_0;\,v_n,\,H) )} \leq e^{\delta (j+1)^2 n v_n^2 } $$
and
$$ \log D( \zeta jv_n,\,\bar S_{n,j} \cap \mathcal F_n,\,\| \cdot \|_1 ) \leq \delta (j+1)^2 n v_n^2 .$$
\end{itemize}
Then, there exists a constant $J_1>0$ such that
$$ \E_{\lambda_0}^{(n)}[\pi(\lambda: \ \| \lambda - \lambda_0\|_{1} >  J_1 v_n \mid N)] = O((n v_n^2)^{-k/2}).$$
\end{thm}

\medskip

The proof of Theorem \ref{th:gene:aalen} is reported in Section \ref{app:proof}.
To the best of our knowledge, the only other paper dealing with posterior concentration rates
in related models is that of \citet{belitser:serra:vanzanten}, where inhomogeneous Poisson processes are considered.
Theorem \ref{th:gene:aalen} differs in two aspects from their Theorem 1.
Firstly, we do not confine ourselves to inhomogeneous Poisson processes. Secondly and more importantly, our
conditions are different: we do not assume that $\lambda_0$ is bounded below away from zero
and we do not need to bound from below the prior mass in neighborhoods of $\lambda_0$ for the sup-norm,
rather the prior mass in neighborhoods of $\lambda_0$ for the Hellinger distance, as in Theorem 2.2 of \citet{ghosal:ghosh:vdv:00}.
In Theorem \ref{th:gene:aalen}, our aim is to propose conditions to assess posterior concentration rates for intensity
functions resembling those used in the density model obtained
by parameterizing $\lambda $ as $\lambda=M_\lambda \times \bar \lambda$, with $\bar \lambda$ a probability density on $\Omega$.
%%%%%%%%%%%%%%%%%%%%%%%%%%%%%%%%%%%%%%%%%%%%%%%%%%%%%%%%%%%%%%%%%%%%%%%%%%%%%%%%%%%%%%%%%%%%%%%%%%%%%%%%%%%%%%%%%%%%%%%%%%%%%%%%%%%%%%%%%%%%%%%%%%%%%
\begin{rmk}
If $\bar \lambda \in \bar B_{2,n}(\bar \lambda_0;\,v_n,\,H)$ then, for every integer $j > 2$,
$ E_j(\bar \lambda_0 ;\, \bar \lambda) \leq H^{j-2} v_n^2 (\log n)^{j-2}$
so that, using Proposition \ref{prop:KL:Aalen}, if we replace $\bar B_{k,n}(\bar \lambda_0;\,v_n,\,H)$ with
$\bar B_{2,n}(\bar \lambda_0;\,v_n,\,H)$ in the assumptions of Theorem \ref{th:gene:aalen},
we obtain the same type of conclusion: for any $k\geq 2$ such that condition \eqref{moment} is satisfied, we have
$$ \E_{\lambda_0}^{(n)} [\pi(\lambda: \ \| \lambda - \lambda_0\|_{1} >  J_1 v_n \mid N)] = O((n v_n^2)^{-k/2} (\log n)^{k(k_{[2]}-2)/2}),$$
with an extra $(\log n)$-term on the right-hand side of the above equality.
\end{rmk}
%%%%%%%%%%%%%%%%%%%%%%%%%%%%%%%%%%%%%%%%%%%%%%%%%%%%%%%%%%%%%%%%%%%%%%%%%%%%%%%%%%%%%%%%%%%%%%%%%%%%%%%%%%%%%%%%%%%%%%%%%%%%%%%%%%%%%%%%%%%%%%%%%%%%%%%

\begin{rmk}
Condition~\eqref{moment} is satisfied for the above considered examples:
it is verified for inhomogeneous Poisson processes since $Y_t=n$ for every $t$. For the censoring model, $Y_t=\sum_{i=1}^n\1_{Z_i\geq t}$.
For every $i=1,\,\ldots,\,n$, we set $V_i=\1_{Z_i\geq t}-\P(Z_1\geq t)$.
Then, for $k\geq 2$,
\[\begin{split}
\E_{\lambda_0}^{(n)}\left[ \left(  \int_\Omega [Y_t-\mu_n(t)]^2 \mathrm{d} t \right)^k  \right] &=\E_{\lambda_0}^{(n)}\left[ \left(  \int_0^T \left(\sum_{i=1}^nV_i\right)^2\mathrm{d} t \right)^k  \right] \\
&\lesssim  \int_0^T\E_{\lambda_0}^{(n)}\left[  \left(\sum_{i=1}^nV_i\right)^{2k}\right] \mathrm{d} t \\ %T^{k-1}
&\lesssim \int_0^T\left(\sum_{i=1}^n \E_{\lambda_0}^{(n)}[V_i^{2k}]+\left(\sum_{i=1}^n\E_{\lambda_0}^{(n)}[V_i^2]\right)^k\right)\mathrm{d} t
\lesssim n^k % C(k,\,T)\\&\,\,\,\,\times
\end{split}\]
by H\"{o}lder and Rosenthal inequalities (see, for instance, Theorem~C.2 of \cite{HKPT}). Under mild conditions, similar computations can be performed for finite state Markov processes.
\end{rmk}

Conditions of Theorem \ref{th:gene:aalen} are very similar to those considered for density estimation
in the case of i.i.d. observations. In particular,
 $$\bar B_n = \left\{\bar \lambda:\,   h^2( \bar \lambda_0,\, \bar \lambda)    \left\| \frac{ \bar \lambda_0 }{\bar \lambda}\right\|_\infty  \leq v_n^2, \,\,  \left\| \frac{ \bar \lambda_0 }{\bar \lambda}\right\|_\infty \leq n^H , \,\, \|\bar \lambda\|_\infty \leq H \right\} $$
 is included in $\bar B_{k,n} \left(\bar \lambda_0;\,v_n (\log n)^{1/2},\, H\right)$
 as a consequence of Theorem 5.1 of \citet{wong:shen:1995}. Apart from the mild constraints
  $\left\| { \bar \lambda_0 }/{\bar \lambda}\right\|_\infty \leq n^H$ and $\|\bar \lambda\|_\infty \leq H$, the set
  $\bar B_n$ is the same as the one considered in Theorem 2.2 of \citet{ghosal:ghosh:vdv:00}.
  The other conditions are essentially those of Theorem 2.1 in \citet{ghosal:ghosh:vdv:00}.
%%%%%%%%%%%%%%%%%%%%%%%%%%%%%%%%%%%%%%%%%%%%%%%%%%%%%%%%%%%%%%%%%%%%%%%%%%%%%%%%%%%%%%%%%%%%%%%%%%%%%%%%%%%

\section{Illustrations with different families of priors}\label{sec:priors}

As discussed in Section \ref{aalen},
the conditions of Theorem \ref{th:gene:aalen} to derive posterior contraction rates
are very similar to those considered in the literature for density estimation so that
existing results involving different families of prior distributions
can be adapted to Aalen multiplicative intensity models. Some applications are presented below.

\subsection{Monotone non-increasing intensity functions}\label{monotone}
In this section, we deal with estimation of monotone non-increasing intensity functions,
which is equivalent to considering monotone non-increasing density functions $\bar\lambda$ in the above described parametrization.
To construct a prior on the set of monotone non-increasing densities over $[0,\,T]$, we use their representation
as mixtures of uniform densities as in \citet{williamson:56} and consider a Dirichlet process as a prior on the mixing distribution:
\begin{equation}\label{DPM:unif}
\bar \lambda(\cdot) = \int_0^\infty \frac{\1_{(0,\,\theta)}(\cdot)}{ \theta} \d P(\theta) , \qquad P \mid A,\, G\sim \textrm{DP}(A G),
\end{equation}
where $G $ is a distribution on $[0,\,T]$ having density $g$ with respect to Lebesgue measure. This prior has been studied by \citet{salomond:13} for estimating monotone non-increasing densities. Here, we extend his results to the case of monotone non-increasing intensity functions of Aalen processes. We consider the same assumption on $G$ as in \citet{salomond:13}:
there exist $a_1,\, a_2 >0$ such that
\begin{equation}\label{ass:G}
\theta^{a_1} \lesssim g(\theta) \lesssim \theta^{a_2} \quad\mbox{ for all $\theta $ in a neighbourhood of $0$.}
\end{equation}
The following result holds.
\begin{cor}\label{cor:aalen}
Assume that the counting process $ N$ verifies conditions \eqref{ass:Y1} and \eqref{ass:Y2} and that inequality \eqref{moment} is satisfied for some $k\geq 2$.
Consider a prior $\pi_1$ on $\bar\lambda$ satisfying conditions \eqref{DPM:unif} and \eqref{ass:G} and a prior $\pi_M$ on $M_\lambda$ that is absolutely continuous with respect to Lebesgue measure with positive and continuous density on $\R_+$. Suppose that $\lambda_0$ is monotone non-increasing and bounded on $\R_+$.
Let $\bar\epsilon_n=(n/\log n)^{-1/3}$. Then, there exists a constant $J_1>0$ such that
$$ \E_{\lambda_0}^{(n)} [\pi(\lambda: \ \| \lambda - \lambda_0\|_{1} > J_1 \bar\epsilon_n\mid N)] = O((n \bar\epsilon_n^2)^{-k/2} (\log n)^{k(k_{[2]}-2)/2}).$$
\end{cor}

The proof is reported in Section \ref{app:proof}.
%%%%%%%%%%%%%%%%%%%%%%%%%%%%%%%%%%%%%%%%%%%%%%%%%%%%%%%%%%%%%%%%%%%%%%%%%%%%%%%%%%%%%%%%%%%%%%%%%%%%%%%%%%%%%%%%%%%%%%%%%%%%%%%%%%%%%%%%%%%%

\subsection{Log-spline and log-linear priors on $ \lambda$} \label{sec:spline}
For simplicity of presentation, we set $T= 1$.  We consider a log-spline prior of order $q$ as in Section 4 of   \citet{ghosal:ghosh:vdv:00}. In other words, $\bar \lambda$ is parameterized as
$$\log \bar \lambda_{\theta} (\cdot) = \theta^t \underline B_{J}(\cdot) - c(\theta), \qquad \mbox{with }\,\exp\left( c(\theta) \right) = \int_0^1 e^{\theta^t \underline B_{J}(x)} \d x, $$
where $\underline B_J = (B_1,\, \ldots,\, B_J)$ is the $q$-th order $B$-spline defined in \citet{deboor:78} associated with $K$ fixed knots, so that $J = K+q-1$, see  \citet{ghosal:ghosh:vdv:00} for more details.
Consider a prior on $\theta$ in the form $J = J_n = \lfloor n^{1/(2\alpha+ 1)}\rfloor$, $\alpha \in [1/2,\, q]$ and, conditionally on $J$, the prior is absolutely continuous with respect to Lebesgue measure on $[-M,M]^J$ with density bounded from below and above by $c^J$ and $C^J$, respectively. Consider an absolutely continuous prior with positive and continuous density on $\R_+$ on $M_\lambda$. We then have the following posterior concentration result.
 \begin{cor} \label{spline}
  For the above prior, if $\|\log \lambda_0 \|_\infty <\infty $ and $\lambda_0$ is H\"older with regularity $\alpha \in [1/2,\, q]$, then under condition \eqref{moment}, there exists a constant $J_1>0$ so that
$$ \E_{\lambda_0}^{(n)} [\pi(\lambda: \ \| \lambda - \lambda_0\|_{1} >  J_1 n^{-\alpha /(2 \alpha + 1)} \mid N)] = O(n^{-k/(4\alpha+ 2)} (\log n)^{k(k_{[2]}-2)/2} ).$$
\end{cor}

\begin{proof}
Set $\epsilon_n = n^{-\alpha/(2\alpha+1)}$.
Using Lemma 4.1, there exists $\theta_0 \in \R^J$ such that $h(\bar \lambda_{\theta_0},\, \bar \lambda_{0}) \lesssim  \| \log \bar \lambda_{\theta_0} - \log \bar \lambda_{0}\|_\infty \lesssim J^{-\alpha}$, which combined with Lemma 4.4 leads to
$$\pi_1(\bar B_{k,n}(\bar \lambda_0;\,\epsilon_n,\,H)\geq e^{-C_1 n\epsilon_n^2}.$$
Lemma 4.5 together with Theorem 4.5 of \citet{ghosal:ghosh:vdv:00} controls the entropy of $\bar S_{n,j}$ and its prior mass for $j$ larger than some fixed constant $J_0$.
\end{proof}

With such families of priors, it is more interesting to work with non-normalized $\lambda_\theta$. We can write
$$\lambda_{A,\theta}(\cdot) = A\exp\left(  \theta^t \underline B_{J}(\cdot) \right), \quad A >0 ,$$
so that a prior on $\lambda$ is defined as a prior on $A$, say $\pi_A$ absolutely continuous with respect to Lebesgue measure having positive and continuous density and the same type of prior  prior on $\theta$ as above.
The same result then holds. It is not a direct consequence of Theorem \ref{th:gene:aalen}, since $M_{\lambda_{A,\theta}} = A \exp( c(\theta) )$ is not \emph{a priori} independent of $\bar \lambda_{A,\theta}$. However, introducing $A$ allows to adapt Theorem \ref{th:gene:aalen} to this case. The practical advantage of the latter representation is that it avoids computing the normalizing constant $c(\theta)$.

In a similar manner, we can replace spline basis with other orthonormal bases, as considered in \citet{rivoirard:rousseau:12}, leading to the same posterior concentration rates as in density estimation. More precisely, consider intensities parameterized as
\begin{equation*}
\bar \lambda_\theta(\cdot)  = e^{\sum_{j=1}^J \theta_j \phi_j(\cdot) -c(\theta) }, \quad e^{c(\theta)} = \int_{\R^J} e^{\sum_{j=1}^J \theta_j \phi_j(x) }\d x,
\end{equation*}
where $(\phi_j)_{j=1}^\infty$ is an orthonormal basis of $\mathbb{L}_2([0,\,1])$, with $\phi_1= 1$. Write $\eta = (A, \,\theta)$, with $A>0$, and
$$\lambda_{\eta }(\cdot) =A  e^{\sum_{j=1}^J \theta_j \phi_j(\cdot)} = Ae^{c(\theta)} \bar \lambda_\theta(\cdot) .$$
Let $A\sim \pi_A$ and consider the same family of priors as in \citet{rivoirard:rousseau:12}:
\begin{equation*}
\begin{split}
J &\sim \pi_J, \\
j^{\beta} \theta_j/\tau_0 &\overset{\textrm{ind}}{\sim} g, \,\,  j \leq J , \quad \mbox{ and }\quad \theta_j =0, \quad \forall\, j >J,
\end{split}
\end{equation*}
where $g$ is a positive and continuous density on $\R$ and there exist $s\geq 0$ and $p>0$ such that
$$ \log \pi_J(J) \asymp - J (\log J)^{s} , \qquad \log g(x) \asymp-|x|^p, \quad s=0,\,1,$$
when $J$ and $|x|$ are large. \citet{rivoirard:rousseau:12} prove that this prior leads to minimax adaptive  posterior concentration rates over collections  of positive and H\"older classes of densities in the density model. Their proof easily extends to prove assumptions $(i)$ and $(ii)$ of Theorem \ref{th:gene:aalen}.
\begin{cor}\label{cor:loglin}
Consider the above described prior on an intensity function $\lambda$ on $[0,\,1]$. Assume that $\lambda_0$ is positive and belongs to a Sobolev class with smoothness $\alpha> 1/2$. Under condition \eqref{moment}, if $\beta< 1/2+\alpha$, there exists a constant $J_1>$ so that
 \[\begin{split}&\E_{\lambda_0}^{(n)} [\pi(\lambda: \ \| \lambda - \lambda_0\|_{1} >  J_1 (n/\log n)^{-\alpha /(2 \alpha + 1)}(\log n)^{(1-s)/2} \mid N)]\\&\hspace*{7cm} = O(n^{-k/(4\alpha+ 2)} (\log n)^{k(k_{[2]}-2)/2} ).
 \end{split}\]
\end{cor}

Note that  the constraint $\beta < \alpha + 1/2$ is satisfied for all $\alpha >1/2$ as soon as $\beta < 1$ and, as in \citet{rivoirard:rousseau:12}, the prior leads to adaptive minimax posterior concentration rates over collections of Sobolev balls.

%%%%%%%%%%%%%%%%%%%%%%%%%%%%%
%%%%%%%%%%%%%%%%%%%%%%%%%%%%%
%%%%%%%%%%%%%%%%%%%%%%%%%%%%%
%%%%%%%%%%%%%%%%%%%%%%%%%%%%%
\section{Proofs}\label{app:proof}
%We denote by $C$ a constant depending on $m_1$, $m_2$, $k$ and so on, which may change from line to line.
%%%%%%%%%%%%%%%%%%%%%%%%%%%%%
%\subsection{Proof of Theorem \ref{th:gene:aalen}} \label{pr:gene:aalen}
%For any intensity $\lambda$, recall that we denote $M_\lambda=\int_\Omega\lambda(t)\d t$ and $\bar\lambda=M_\lambda^{-1}\times\lambda\in{\mathcal F}_1.$

To prove Theorem \ref{th:gene:aalen}, we use the following intermediate results whose proofs are postponed
to Section \ref{app:Aalen}.
%that are based on classical tools subsequently defined.
The first one controls the Kullback-Leibler divergence and absolute moments of $\ell_n(\lambda_0)-\ell_n(\lambda)$, where $\ell_n(\lambda)$ is the log-likelihood for Aalen processes evaluated at $\lambda$, whose expression is given by
\begin{equation*}\label{loglik:aalen}
\ell_n(\lambda)=\int_0^T\log(\lambda(t))\d N_t-\int_0^T\lambda(t)Y_t\d t,
\end{equation*}
see \cite{ABGK}.

%For Aalen models, we have for $k\geq 1$,
%$$\textrm{KL}(\lambda_0;\,\lambda)=\E_{\lambda_0}^{(n)}[\ell_n(\lambda_0)-\ell_n(\lambda)],\quad V_{k}(\lambda_0;\,\lambda) = \E_{\lambda_0}^{(n)}[|\ell_n(\lambda_0)-\ell_n(\lambda)-\E_{\lambda_0}^{(n)}[\ell_n(\lambda_0)-\ell_n(\lambda)]|^{k}].$$

\begin{prop} \label{prop:KL:Aalen}
Let $v_n$ be a positive  sequence such that $v_n\rightarrow0$ and $nv_n^2 \rightarrow \infty$.
For any $k\geq 2$ and $ H>0$, define the set
$$B_{k,n}(\lambda_0;\, v_n, \, H) =
\{ \lambda:  \ \bar \lambda \in \bar B_{k,n}(\bar \lambda_0;\, v_n,\,H), \,\,\, |M_\lambda - M_{\lambda_0}| \leq v_n\}.$$
Under assumptions \eqref{ass:Y1} and \eqref{moment}, for all $\lambda \in B_{k,n}(\lambda_0;\,v_n,\,H) $, we have
\begin{equation*}
\mathrm{KL}(\lambda_0;\,\lambda)\leq \kappa_0n v_n^2 \quad\mbox{ and }\quad V_k(\lambda_0;\, \lambda)\leq \kappa( n v_n^2)^{k/2},
\end{equation*}
where $\kappa_0,\, \kappa$ depend only on $k$, $C_{1k}$, $H$, $\lambda_0$, $m_1$ and $m_2$. An expression of $\kappa_0$ is given in \eqref{kap0}.
\end{prop}
%%%%%%%%%%%%%%%%%%%%%%%%%%%%%%%%%%%%%%%%%%%%%%%%%%%%%%%%%%%%%%%%%%%%%%%%%%%%%%%%%%%%%%%%%%%%%%%%%%%%%%%%%%%%%%%%%
The second result establishes the existence of tests that are used to control the numerator of posterior distributions.
We use that, under assumption \eqref{ass:Y1}, on the set $\Gamma_n$,
\begin{equation}\label{ass:Y3}
\forall\, t\in \Omega,\quad (1-\alpha)\tilde\mu_n(t) \leq  \frac{Y_t}{n} \leq (1+\alpha) \tilde\mu_n(t).
\end{equation}
\begin{prop} \label{prop:test:Aalen}
Assume that conditions $(i)$ and $(ii)$ of Theorem \ref{th:gene:aalen} are satisfied. For any $j\in\mathbb{N}$, define
$$S_{n,j}(v_n) = \{ \lambda: \ \bar\lambda\in{\mathcal F}_n\,\mbox{ and }\, jv_n < \|\lambda -\lambda_0\|_1 \leq (j+1)v_n \}.$$
Then, under assumption \eqref{ass:Y1},
there are constants $J_0 ,\, \rho,\, c>0$ such that, for every integer $j\geq J_0$, there exists a  test
$\phi_{n,j}$
%:\,\Xn\rightarrow [0,\,1]$
so that, for a positive constant $C$,
\[
  \begin{array}{lll}
    \E_{\lambda_0}^{(n)}[\mathbf{1}_{\Gamma_n} \phi_{n,j}] \leq Ce^{ - c n j^2 v_n^2 }, \,
\sup\limits_{\lambda \in S_{n,j}(v_n) }\E_{\lambda}[\mathbf{1}_{\Gamma_n} (1- \phi_{n,j}) ] \leq Ce^{ - c n j^2 v_n^2 }, & \hbox{$J_0 \leq j \leq \dfrac{\rho}{v_n}$,} \\[7pt]
\mbox{and}\\[7pt]
    \E_{\lambda_0}^{(n)}[\mathbf{1}_{\Gamma_n} \phi_{n,j}] \leq Ce^{ - c n j v_n }, \,\,\,\,\,\,\,
\sup\limits_{\lambda \in S_{n,j}(v_n) }\E_{\lambda}[\mathbf{1}_{\Gamma_n} (1- \phi_{n,j}) ] \leq Ce^{ - c n j v_n }, & \hbox{$\qquad\, j > \dfrac{\rho}{v_n}$.}
  \end{array}\]
%$$ \E_{\lambda_0}^{(n)}[\mathbf{1}_{\Gamma_n} \phi_{n,j}] \leq Ce^{ - c n (j v_n)^2 }, \,\,\,
%\sup_{\lambda \in S_{n,j}(v_n) }\E_{\lambda}[\mathbf{1}_{\Gamma_n} (1- \phi_{n,j}) ] \leq Ce^{ - c n (j v_n)^2 },\,\,\,\mbox{ }  j \leq %\frac{\rho}{v_n},$$
%$$ \E_{\lambda_0}^{(n)}[1_{\Gamma_n} \phi_{n,j}] \leq Ce^{ - c n j v_n }, \,\,\,
%\sup_{\lambda \in S_{n,j}(v_n) }\E_{\lambda}[1_{\Gamma_n} (1- \phi_{n,j}) ] \leq Ce^{ - c n j v_n },\,\,\,\mbox{ }  j \geq \frac{\rho}{v_n},$$
%for a positive constant $C$.
\end{prop}

\smallskip

In what follows, the symbols \vir{$\lesssim$} and \vir{$\gtrsim$}
are used to denote inequalities valid up to constants that are universal or
fixed throughout.

\begin{proof}[Proof of Theorem \ref{th:gene:aalen}]% \label{pr:gene:aalen}
Given Proposition~\ref{prop:KL:Aalen} and Proposition \ref{prop:test:Aalen},
the proof of Theorem \ref{th:gene:aalen} is similar to that of Theorem 1 in \citet{ghosal:vdv:07}.
Let $U_n = \{\lambda: \ \| \lambda - \lambda_0\|_1 > J_1v_n \}$.
Write
\begin{equation*}
\pi(U_n\mid N) = \dfrac{ \int_{U_n} e^{\ell_n(\lambda) - \ell_n(\lambda_0) }\d\pi(\lambda) }{ \int_{\mathcal F} e^{\ell_n(\lambda) - \ell_n(\lambda_0) }\d\pi(\lambda) }=\dfrac{N_n}{D_n}.
\end{equation*}
%where $\ell_n(\lambda)$ is the log-likelihood evaluated at $\lambda$.
%We use the notation of Proposition~\ref{prop:KL:Aalen} and Proposition~\ref{prop:test:Aalen}. Writing
% $$ D_n = \int_{\mathcal F} e^{\ell_n(\lambda) - \ell_n(\lambda_0) }\d\pi(\lambda),$$
We have
\[
 \begin{split}
 &
 \P_{\lambda_0}^{(n)}\left( D_n \leq e^{ - (\kappa_0+1) n v_n^2 }\pi_1( \bar B_{k,n}(\bar \lambda_0;  v_n,\,H) ) \right) \\
& \quad\qquad\leq \P_{\lambda_0}^{(n)}\left(\int_{B_{k,n}(\lambda_0;\,v_n,\, H)}\frac{ \exp\{\ell_n(\lambda) - \ell_n(\lambda_0)\}}{\pi(B_{k,n}(\lambda_0;\,
v_n,\,H) ) }\d\pi(\lambda)\right. \\&\hspace*{5cm}\left.\leq - (\kappa_0+1) n v_n^2+\log\left(\frac{\pi_1( \bar B_{k,n}(\bar \lambda_0;\, v_n,\,H))}{\pi(B_{k,n}(\lambda_0;\,  v_n,\,H))}\right)\right).
\end{split}
\]
By the assumption on the positivity and continuity of the Lebesgue density
of the prior $\pi_M$ and the requirement that $v_n^2 \geq (n/\log n)^{-1}$,
 $$ \pi(B_{k,n}(\lambda_0;\, v_n,\, H))\gtrsim \pi_1( \bar B_{k,n}(\bar \lambda_0;\,v_n,\,H))v_n
\gtrsim \pi_1( \bar B_{k,n}(\bar \lambda_0;\, v_n,\,H) )e^{-nv_n^2/2},  $$ so that,
using Proposition \ref{prop:KL:Aalen} and Markov's inequality,
 $$  \P_{\lambda_0}^{(n)}\left( D_n \leq e^{ - (\kappa_0 +1) n  v_n^2 }\pi_1( \bar B_{k,n}(\bar \lambda_0;\, v_n,\, H) ) \right) \lesssim (nv_n^2)^{-k/2}.$$
Note that inequality \eqref{norm:mino} implies that $\pi(S_{n,j}(v_n))\leq \pi_1(\bar S_{n,j})$.
Using tests $\phi_{n,j}$ of Proposition \ref{prop:test:Aalen}, mimicking the proof of Theorem~1 of \cite{ghosal:vdv:07},
we have that for  $J_1 \geq J_0$,
 \begin{equation*}
 \begin{split}
   &\hspace*{-0.5cm}\E_{\lambda_0}^{(n)} \left[ \1_{\Gamma_n} \pi\left( \lambda : \ \|\lambda -\lambda_0\|_1> J_1 v_n \mid N\right) \right]\\
   & \quad\hspace*{0.8cm} \leq \sum_{j\geq J_1 }  \E_{\lambda_0}^{(n)} [\1_{\Gamma_n}\phi_{n,j} ]+ \sum_{j = \lceil J_1\rceil}^{\lfloor \rho/v_n\rfloor }  e^{  (\kappa_0 +1) n v_n^2 } \frac{ \pi_1(\bar S_{n,j}) e^{ - c n j^2 v_n^2 } }{ \pi_1( \bar B_{k,n}(\bar \lambda_0;\,v_n,\,H) )  } \\
 &\quad\hspace*{2.8cm}  + \sum_{j > \rho/v_n  }   \frac{ e^{  (\kappa_0 +1) n v_n^2 }\pi_1(\bar S_{n,j} ) e^{ - c n j v_n } }{ \pi_1( \bar B_{k,n}(\bar \lambda_0;\,v_n,\,H) )  } + \frac{e^{  (\kappa_0 +1) n  v_n^2 }  \pi_1(\mathcal F_n^c )  }{ \pi_1( \bar B_{k,n}(\bar \lambda_0;\,v_n,\,H) )  }\\
  & \quad \hspace*{2.8cm} + \P_{\lambda_0}^{(n)}( D_n \leq e^{ - (\kappa_0 +1) n v_n^2 }\pi_1( \bar B_{k,n}(\bar \lambda_0;\,v_n,\,H) ) )\\
  & \quad\hspace*{0.8cm} \lesssim (nv_n^2)^{-k/2},
 \end{split}
 \end{equation*}
which proves the result since $\P_{\lambda_0}^{(n)}(\Gamma_n^c)=o(1)$.
\end{proof}

\medskip

%%%%%
\begin{proof}[Proof of Corollary \ref{cor:aalen}] %\label{pr:cor:aalen}
Without loss of generality, we can assume that $\Omega=[0,\,T]$.
At several places, using \eqref{informally} and \eqref{ass:Y3}, we have that, under $\P_\lambda^{(n)}(\cdot\mid \Gamma_n)$, for any interval $I$, the number of points of $N$ falling in $I$ is controlled by the number of points of a Poisson process with intensity $n(1+\alpha)m_2\lambda$ falling in $I$.
Recall that $\bar\epsilon_n = (n/\log n)^{-1/3} $.
%and $\epsilon_n = J_1 \bar \epsilon_n$ for some constant $J_1>0$ large enough.
For $\kappa_0$ as in \eqref{kap0}, we control $\P_{\lambda_0}^{(n)}( \ell_n(\lambda) - \ell_n( \lambda_0) \leq - (\kappa_0+1)n \bar \epsilon_n^2 ) $. We follow most of the computations of \citet{salomond:13}. Let $e_n =  (n \bar\epsilon_n^2)^{-k/2}$,
$$\bar \lambda_{0n} (t) =  \frac{ \lambda_0(t)  \1_{ t \leq \theta_n} }{ \int_0^{\theta_n} \lambda_0(u) \d u },  \quad \mbox{ with }\,\, \theta_n = \inf \left\{ \theta:\, \int_{0}^\theta \bar \lambda_0(t) \d t \geq 1 - \frac{e_n  }{n } \right\},$$ and $\lambda_{0n} = M_{\lambda_0}\bar \lambda_{0n}$.
Define the event $A_n = \{X\in N: \, X \leq \theta_n \}$. We make use of the following result. Let $\tilde N$ be a Poisson process with intensity $n(1+\alpha)m_2\lambda_0$. If $\tilde N(T)=k$, denote by $\tilde N=\{X_1,\,\ldots,\,X_k\}$. Conditionally on $\tilde N(T)=k$, the random variables
$X_1,\,\ldots,\,X_k$ are i.i.d. with density $\bar\lambda_0$. So,
\[
\begin{split}
\P_{\lambda_0}^{(n)}(A_n^c\mid \Gamma_n)&\leq \sum_{k=1}^{\infty}\P_{\lambda_0}^{(n)}(\exists\; X_i>\theta_n\mid \tilde N(T)=k)\,\P_{\lambda_0}^{(n)}(\tilde N(T)=k)\\
&\leq \sum_{k=1}^{\infty}\left(1-\left(1-\frac{e_n}{n}\right)^k\right)\P_{\lambda_0}^{(n)}(\tilde N(T)=k)\\
&= O\left(\frac{e_n}{n}\E_{\lambda_0}^{(n)}[\tilde N(T)]\right)=O(e_n)= O((n \bar\epsilon_n^2)^{-k/2}).
\end{split}
\]
Now,
\[\begin{split}
&\P_{\lambda_0}^{(n)}\left( \ell_n(\lambda) - \ell_n( \lambda_0) \leq - (\kappa_0+2)n \bar \epsilon_n^2\mid \Gamma_n \right)\\
&\hspace*{2cm}\leq \P_{\lambda_0}^{(n)}\left( \ell_n(\lambda) - \ell_n( \lambda_0) \leq - (\kappa_0+2)n\bar  \epsilon_n^2\mid A_n,\,\Gamma_n \right)+ \P_{\lambda_0}^{(n)}(A_n^c\mid \Gamma_n).
\end{split}\]
We now deal with the first term on the right-hand side. On $\Gamma_n\cap A_n$,
\begin{equation*}
\begin{split}
\ell_n (\lambda_0) &= \ell_n( \lambda_{0n}) + \int_0^{\theta_n}\log\left(\frac{\lambda_0(t)}{\lambda_{0n}(t)}\right)\d N_t-\int_0^T[\lambda_0(t)-\lambda_{0n}(t)]Y_t\d t
\\
&=\ell_n( \lambda_{0n}) +N(T) \log \left( \int_0^{\theta_n} \bar \lambda_0(t) \d t \right)  - M_{\lambda_0} \int_0^T \bar \lambda_0(t)Y_t  \d t + M_{\lambda_0}\frac{ \int_0^{\theta_n} \bar \lambda_0(t)Y_t  \d t}{ \int_0^{\theta_n} \bar \lambda_0(t) \d t }  \\
&\leq\ell_n( \lambda_{0n}) + M_{\lambda_0}\frac{ \int_{\theta_n}^T  \bar \lambda_0(t) \d t \int_0^{\theta_n} \bar \lambda_0(t)Y_t  \d t}{ \int_0^{\theta_n} \bar \lambda_0(t) \d t}
 -  M_{\lambda_0} \int_{\theta_n}^T \bar \lambda_0(t)Y_t  \d t \\
 &\leq \ell_n( \lambda_{0n}) + M_{\lambda_0} \frac{e_n (1+\alpha)m_2 }{ 1 - e_n/n }.
\end{split}
\end{equation*}
So, for every $\lambda$ and any $n$ large enough,
\[\begin{split}
&\P_{\lambda_0}^{(n)}\left( \ell_n(\lambda) - \ell_n( \lambda_0) \leq - (\kappa_0+2)n \bar \epsilon_n^2\mid A_n,\,\Gamma_n \right)\\
&\hspace*{4.5cm}\leq
\P_{\lambda_0}^{(n)}\left( \ell_n(\lambda) - \ell_n( \lambda_{0n}) \leq - (\kappa_0+1)n \bar  \epsilon_n^2\mid A_n,\,\Gamma_n \right)\\
&\hspace*{4.5cm}=\P_{\lambda_{0n}}^{(n)} \left( \ell_n(\lambda) - \ell_n( \lambda_{0n})  \leq - (\kappa_0+1)n \bar \epsilon_n^2\mid \Gamma_n \right)
\end{split}\]
because $\P_{\lambda_0}^{(n)}(\cdot\mid A_n)=\P_{\lambda_{0n}}^{(n)}(\cdot).$
Let $H>0$ be fixed. For all $\lambda\in B_{k,n}(\lambda_{0n};\,\bar\epsilon_n,\,H)$, using Proposition \ref{prop:KL:Aalen}, we obtain
\begin{equation*} \label{likeratio:aalen}
\P_{\lambda_{0n}}^{(n)} \left( \ell_n(\lambda) - \ell_n(\lambda_{0n}) \leq -(\kappa_0+1)n \bar \epsilon_n^2\mid \Gamma_n\right) = O((n\bar\epsilon_n^2)^{-k/2}).
\end{equation*}
Mimicking the proof of Lemma 8 in \citet{salomond:13}, we have that, for some constant $C_k >0$,
\begin{equation*}\label{salomond:lem8}
  \pi_1 \left( \bar B_{k,n}(\bar \lambda_{0n}; \,\bar \epsilon_n,\,H)\right) \geq e^{ - C_k n \bar \epsilon_n^2 }\quad\mbox{ when $n$ is large enough,}
 \end{equation*}
so that the first part of condition $(ii)$ of Theorem \ref{th:gene:aalen} is verified.
%\textcolor[rgb]{0.98,0.00,0.00}{as long as $J_1$ is large enough ??? do we need this?}.
As in \citet{salomond:13}, we set $\mathcal F_n = \{ \bar \lambda: \ \bar \lambda (0 ) \leq M_n\}$, with
$M_n = \exp ( c_1 n \bar\epsilon_n^2 )$ and $c_1$ a positive constant. From Lemma 9 of \citet{salomond:13}, there exists $a>0$ such that
$  \pi_1(\mathcal F_n^c) \leq e^{ - c_1 (a+1) n\bar\epsilon_n^2 } $ for $n$ large enough, and the first part of  condition $(i)$ is satisfied.
It is known from \citet{groeneboom:85} that the $\epsilon$-entropy of $\mathcal F_n$ is of order $(\log M_n)/\epsilon$, that is $o(n)$ for all $\epsilon>0$ and the second part of $(i)$ holds.
The second part of $(ii)$ is a consequence of equation (22) of  \citet{salomond:13}.
\end{proof}

%%%%%%%%%%%%%%%%%%%%%%%%%%%%%%%%%%%%%%%%%%%%%%%%%%%%%%%%%%%%%%%%%%%%%%%%%%%%%%%%%%%%%%%%%%%%%%%%

%\appendix

%%%%%%%%%%%%%%%%
%%%%%%%%%%%%%%%%
%\appendix
\section{Auxiliary results}\label{app:Aalen}
This section reports the proofs of Proposition \ref{prop:KL:Aalen} and Proposition
\ref{prop:test:Aalen} that have been stated in Section \ref{app:proof}.
Proofs of intermediate results are deferred to Section \ref{app}.

We use the fact that for any pair of densities $f$ and $g$,
$\|f-g\|_1\leq 2h(f,\,g).$
%%%%%%%%%%%%%%%%%%%%%%%

\begin{proof}[Proof of Proposition \ref{prop:KL:Aalen}]
%We use standard properties of continuous time martingales (see Appendix B of Karr for instance).
%\red{the following weighted norm
%$$\| \lambda - \tilde\lambda \|_{\tilde\mu_n}  = \int_\Omega | \lambda(t)  - \tilde\lambda(t)  |\tilde\mu_n(t) dt, $$
%while and
%$$M_{0,n}:=\int_0^T\lambda_0(t)\mu_n(t)\d t.$$}
Recall that the log-likelihood evaluated at $\lambda$ is given by
$\ell_n(\lambda)=\int_0^T\log(\lambda(t))\d N_t-\int_0^T\lambda(t)Y_t\d t$.  Since on $\Omega^c$, $N$ is empty and $Y_t\equiv0$ almost surely, we can assume, without loss of generality, that $\Omega=[0,\,T]$.
Define $$M_n(\lambda) = \int_0^T \lambda(t) \mu_n(t) \d t,\quad\mbox{} \quad M_n(\lambda_0) = \int_0^T \lambda_0(t) \mu_n(t) \d t,$$
and
$$\bar \lambda_n(\cdot) = \frac{\lambda(\cdot) \mu_n(\cdot) }{M_n(\lambda)}= \frac{ \bar \lambda(\cdot) \tilde\mu_n(\cdot)}{ \int_0^T \bar \lambda(t) \tilde\mu_n(t)\d t}, \quad \mbox{} \quad  \bar \lambda_{0,n}(\cdot) = \frac{ \lambda_0(\cdot) \mu_n(\cdot) }{M_n(\lambda_0) } = \frac{ \bar \lambda_0(\cdot)\tilde\mu_n(\cdot)}{ \int_0^T \bar \lambda_0(t)\tilde \mu_n(t)\d t}.$$
By straightforward computations,
\begin{equation}\label{compil1}
\begin{split}
\textrm{KL}(\lambda_0;\,\lambda)&=\E_{\lambda_0}^{(n)}[\ell_n(\lambda_0)-\ell_n(\lambda)]\\
&=M_{n}(\lambda_0)\pq{\textrm{KL}
(\bar{\lambda}_{0,n};\,\bar{\lambda}_{n})+\frac{M_n(\lambda)}{M_{n}(\lambda_0)} -1-\log\left(\frac{M_n(\lambda)}{M_{n}(\lambda_0)}\right)}\\
&=M_{n}(\lambda_0)\pq{\textrm{KL}(\bar{\lambda}_{0,n};\,\bar{\lambda}_{n})+\phi\left(\frac{M_n(\lambda)}{M_{n}(\lambda_0)}\right)}\\
&\leq nm_2M_{\lambda_0}\pq{\textrm{KL}(\bar{\lambda}_{0,n};\,\bar{\lambda}_{n})+\phi\left(\frac{M_n(\lambda)}{M_{n}(\lambda_0)}\right)},
\end{split}
\end{equation}
where $\phi(x) = x- 1 - \log x$ and
$$\textrm{KL}(\bar{\lambda}_{0,n};\,\bar{\lambda}_{n})=
\int_0^T\log\left(\frac{\bar\lambda_{0,n}(t)}{\bar\lambda_{n}(t)}\right)\bar\lambda_{0,n}(t)\d t.$$
We control $\textrm{KL}(\bar{\lambda}_{0,n};\,\bar{\lambda}_{n})$ for $\lambda \in B_{k,n}(\lambda_0;\,v_n,\,H)$. By using Lemma~8.2 of \cite{ghosal:ghosh:vdv:00}, we have
\begin{eqnarray}\begin{split}\label{compil2}
\textrm{KL}(\bar{\lambda}_{0,n};\,\bar{\lambda}_{n})&\leq 2h^2(\bar{\lambda}_{0,n},\,\bar{\lambda}_{n})\left(1+\log\left\|\frac{\bar{\lambda}_{0,n}}{\bar{\lambda}_{n}}\right\|_\infty\right)\\
&\leq2h^2(\bar{\lambda}_{0,n},\,\bar{\lambda}_{n})
\pq{1+\log\left(\frac{m_2}{m_1}\right)+\log\left\|\frac{\bar{\lambda}_0}{\bar{\lambda}}\right\|_\infty}\\
&\leq2\pq{1+\log\left(\frac{m_2}{m_1}\right)}h^2(\bar{\lambda}_{0,n},\,\bar{\lambda}_{n})
\left(1+\log\left\|\frac{\bar{\lambda}_0}{\bar{\lambda}}\right\|_\infty\right)
\end{split}\end{eqnarray}
because $1+\log(m_2/m_1)\geq 1$. We now deal with $h^2(\bar{\lambda}_{0,n},\,\bar{\lambda}_{n})$. We have
\begin{eqnarray*}
\begin{split}
h^2(\bar{\lambda}_{0,n},\,\bar{\lambda}_{n})&=\int_0^T\left(\sqrt{\bar \lambda_{0,n}(t)}-\sqrt{ \bar \lambda_{n}(t)}\right)^2\d t\\
&=\int_0^T\left(\sqrt{\frac{\bar \lambda_0(t)\tilde\mu_n(t)}{\int_0^T\bar\lambda_0(u)\tilde\mu_n(u)\d u}}-\sqrt{\frac{\bar \lambda(t)\tilde\mu_n(t)}{\int_0^T\bar\lambda(u)\tilde\mu_n(u)\d u}}\right)^2\d t\\
&\leq2m_2\int_0^T\left(\sqrt{\frac{\bar \lambda_0(t)}{\int_0^T\bar\lambda_0(u)\tilde\mu_n(u)\d u}}-\sqrt{\frac{\bar \lambda_0(t)}{\int_0^T\bar\lambda(u)\tilde\mu_n(u)\d u}}\right)^2\d t\\&\qquad \qquad\qquad +2m_2\int_0^T\left(\sqrt{\frac{\bar \lambda_0(t)}{\int_0^T\bar\lambda(u)\tilde\mu_n(u)\d u}}-\sqrt{\frac{\bar \lambda(t)}{\int_0^T\bar\lambda(u)\tilde\mu_n(u)\d u}}\right)^2\d t\\
&\leq2m_2U_n+\frac{2m_2}{m_1}h^2(\bar\lambda_0,\,\bar\lambda),
\end{split}
\end{eqnarray*}
with $$U_n=\left(\sqrt{\frac{1}{\int_0^T\bar\lambda_0(t)\tilde\mu_n(t)\d t}}-\sqrt{\frac{1}{\int_0^T\bar\lambda(t)\tilde\mu_n(t)\d t}}\right)^2.$$
We denote by
$$\tilde\epsilon_n:=\frac{1}{\int_0^T\bar\lambda_0(u)\tilde\mu_n(u)\d u}\int_0^T[\bar\lambda(t)-\bar\lambda_0(t)]\tilde\mu_n(t)\d t,$$
so that
$$|\tilde\epsilon_n|\leq\frac{1}{m_1}\int_0^T|\bar\lambda(t)-\bar\lambda_0(t)|\tilde\mu_n(t)\d t\leq\frac{2m_2}{m_1}h(\bar\lambda_0,\,\bar\lambda).$$
Then,
$$U_n=\frac{1}{\int_0^T\bar\lambda_0(t)\tilde\mu_n(t)\d t}\left(1-\frac{1}{\sqrt{1+\tilde\epsilon_n}}\right)^2\leq \frac{\tilde\epsilon_n^2}{4m_1}\leq \frac{m_2^2}{m_1^3}h^2(\bar\lambda_0,\,\bar\lambda).$$
Finally,
\begin{equation}\label{compil3}
h^2(\bar{\lambda}_{0,n},\,\bar{\lambda}_{n})\leq \frac{2m_2}{m_1}\left(\frac{m_2^2}{m_1^2}+1\right)h^2(\bar\lambda_0,\,\bar\lambda).
\end{equation}
It remains to bound $\phi\left(M_n(\lambda)/M_{n}(\lambda_0)\right)$. We have
\begin{equation*}
\begin{split}
|M_n(\lambda_0)-M_n(\lambda)|&\leq\int_0^T|\lambda(t)-\lambda_0(t)|\mu_n(t)\d t\\
&  \leq nm_2\int_0^T|\lambda(t)-\lambda_0(t)|\d t\\
&\leq\frac{m_2}{m_1M_{\lambda_0}}M_n(\lambda_0)\left[M_{\lambda_0}\|\bar\lambda-\bar\lambda_0\|_1+|M_\lambda-M_{\lambda_0}|\right]\\
&\leq\frac{m_2}{m_1M_{\lambda_0}}M_n(\lambda_0)[2M_{\lambda_0}h(\bar\lambda,\,\bar\lambda_0)+|M_\lambda-M_{\lambda_0}|]\\
&\leq\frac{m_2}{m_1M_{\lambda_0}}M_n(\lambda_0)(2M_{\lambda_0}+1)v_n.
\end{split}
\end{equation*}
Since $\phi(u+1)\leq u^2$ if $|u|\leq{1}/{2}$, we have
\begin{equation}\label{compil4}
\phi\left(\frac{M_n(\lambda)}{M_{n}(\lambda_0)}\right)\leq \frac{m_2^2}{m_1^2M_{\lambda_0}^2}(2M_{\lambda_0}+1)^2v_n^2 \quad\mbox{for $n$ large enough.}
\end{equation}
Combining \eqref{compil1}, \eqref{compil2}, \eqref{compil3} and \eqref{compil4}, we have $\textrm{KL}(\lambda_0;\,\lambda)\leq \kappa_0 nv_n^2$  for $n$ large enough, with $\kappa_0$ as in \eqref{kap0}.
%\begin{equation*}\label{kappa0}
%\kappa_0=m_2^2M_{\lambda_0}\left(\frac{4}{m_1}\left(1+\log\left(\frac{m_2}{m1}\right)\right)\left(1+\frac{m_2^2}{m_1^2}\right)+\frac{m_2(2M_{\lambda_0}+1)^2}{m_1^2M_{\lambda_0}^2}\right).
%\end{equation*}
We now deal with
$$ V_{2k}(\lambda_0;\,\lambda) = \E_{\lambda_0}^{(n)}[|\ell_n(\lambda_0)-\ell_n(\lambda)-\E_{\lambda_0}^{(n)}[\ell_n(\lambda_0)-\ell_n(\lambda)]|^{2k}],\quad\mbox{$k\geq1$}.$$
We begin by considering the case $k>1$.
In the sequel, we denote by $C$ a constant that may change from line to line. Straightforward computations lead to
\begin{eqnarray*}
\begin{split}
V_{2k}(\lambda_0;\, \lambda)&=\E_{\lambda_0}^{(n)}\left[\left|-\int_0^T\pq{\lambda_0(t)-\lambda(t)-\lambda_0(t)\log\left(\frac{\lambda_0(t)}{\lambda(t)}\right)}[ Y_t-\mu_n(t)]\d t \right. \right.\\
& \hspace{5cm} \left. \left.+\int_0^T\log\left(\frac{\lambda_0(t)}{\lambda(t)}\right)[\d N_t-Y_t\lambda_0(t)\d t]\right|^{2k}\right]\\
&\leq 2^{2k-1}(A_{2k}+B_{2k}),
\end{split}
\end{eqnarray*}
with
$$B_{2k}:=\E_{\lambda_0}^{(n)}\left[ \left|\int_0^T\log\left(\frac{\lambda_0(t)}{\lambda(t)}\right)[\d N_t-Y_t\lambda_0(t)\d t]\right|^{2k}\right]$$
and, by \eqref{moment},
\begin{eqnarray*}
\begin{split}
A_{2k}&:=\E_{\lambda_0}^{(n)}\left[\left|\int_0^T\pq{\lambda_0(t)-\lambda(t)-\lambda_0(t)\log\left(\frac{\lambda_0(t)}{\lambda(t)}\right)}[Y_t-\mu_n(t)]\d t \right|^{2k} \right]\\
&\leq\left(\int_0^T \pq{\lambda_0(t)-\lambda(t)-\lambda_0(t)\log\left(\frac{\lambda_0(t)}{\lambda(t)}\right)}^2 \d t \right)^k\times
\E_{\lambda_0}^{(n)}\left[ \left(  \int_0^T [Y_t-\mu_n(t)]^2 \d t\right)^k  \right] \\
&\leq 2^{2k-1}C_{1k}n^k\left(A_{2k,1}+A_{2k,2}\right),
\end{split}
\end{eqnarray*}
where, for $\lambda\in B_{k,n}(\lambda_0;\,v_n,\,H)$,
\begin{eqnarray*}
\begin{split}
A_{2k,1}&:=\pq{\int_0^T \lambda_0^2(t)\log^2\left(\frac{\lambda_0(t)}{\lambda(t)}\right) \d t }^k\\
&\leq M_{\lambda_0}^{2k}\|\bar\lambda_0\|_\infty^k\pq{\int_0^T\bar \lambda_0(t)\log^2\left(\frac{M_{\lambda_0}\bar\lambda_0(t)}{M_{\lambda}\bar\lambda(t)}\right) \d t}^k\\
&\leq 2^{2k-1} M_{\lambda_0}^{2k}\|\bar\lambda_0\|_\infty^k\pq{E_2^k(\bar\lambda_0;\,\bar\lambda)+\left|\log\left(\frac{M_\lambda}{M_{\lambda_0}}\right)\right|^{2k}}\\
&\leq C\pq{E_2^k(\bar\lambda_0;\,\bar\lambda)+\left|M_\lambda-M_{\lambda_0}\right|^{2k}}
\leq Cv_n^{2k}
\end{split}
\end{eqnarray*}
and
\begin{eqnarray*}
\begin{split}
A_{2k,2}&:=\left(\int_0^T [\lambda_0(t)-\lambda(t)]^2 \d t \right)^k\\
&=\left(\int_0^T \pg{(M_{\lambda_0}-M_\lambda)\bar\lambda_0(t)-M_\lambda[\bar\lambda(t)-\bar\lambda_0(t)]}^2 \d t \right)^k\\
&\leq 2^{2k-1}\|\bar\lambda_0\|_\infty^{2k}(M_{\lambda_0}-M_\lambda)^{2k}\\
&\qquad\qquad\qquad\quad +2^{2k-1}M_{\lambda}^{2k}
\pq{\int_0^T \left(\sqrt{\bar\lambda_0(t)}-\sqrt{\bar\lambda(t)}\right)^2 \left(\sqrt{\bar\lambda_0(t)}+\sqrt{\bar\lambda(t)}\right)^2 \d t}^k\\
&\leq 2^{2k-1}\|\bar\lambda_0\|_\infty^{2k}(M_{\lambda_0}-M_\lambda)^{2k}+2^kM_{\lambda}^{2k}(\|\bar\lambda_0\|_\infty+\|\bar\lambda\|_\infty)^kh^{2k}(\bar\lambda_0,\,\bar\lambda)
\leq Cv_n^{2k}.
\end{split}
\end{eqnarray*}
Therefore,
$$A_{2k}\leq C(nv_n^2)^k.$$
To deal with $B_{2k}$, for any $T>0$, we set
$$M_T:=\int_0^T\log\left(\frac{\lambda_0(t)}{\lambda(t)}\right)[\d N_t-Y_t\lambda_0(t)\d t],$$
so $(M_T)_T$ is a martingale. Using the Burkholder-Davis-Gundy Inequality (see Theorem B.15 in \cite{Karr}), there exists a constant $C(k)$ only depending on $k$ such that, since $2k> 1$,
$$\E_{\lambda_0}^{(n)}[|M_T|^{2k}]\leq C(k)\E_{\lambda_0}^{(n)}\left[ \left|\int_0^T\log^2\left(\frac{\lambda_0(t)}{\lambda(t)}\right)\d N_t\right|^k\right].$$
Therefore, for $k>1$,
\begin{eqnarray*}
\begin{split}
B_{2k}&= \E_{\lambda_0}^{(n)}[|M_T|^{2k}]\\
&\leq 3^{k-1}C(k)\left(\E_{\lambda_0}^{(n)}\left[ \left|\int_0^T\log^2\left(\frac{\lambda_0(t)}{\lambda(t)}\right)[\d N_t-Y_t\lambda_0(t)\d t]\right|^k\right.\right.\\
&\hspace{6cm}+\left|\int_0^T\log^2\left(\frac{\lambda_0(t)}{\lambda(t)}\right)[Y_t-\mu_n(t)]\lambda_0(t)\d t\right|^k\\
&\hspace{6cm}+\left.\left.\left|\int_0^T\log^2\left(\frac{\lambda_0(t)}{\lambda(t)}\right)\mu_n(t)\lambda_0(t)\d t\right|^k\right]\right)\\
&= 3^{k-1}C(k)(B_{k,2}^{(0)}+B_{k,2}^{(1)}+B_{k,2}^{(2)}),
\end{split}
\end{eqnarray*}
with
\begin{eqnarray*}
\begin{split}
B_{k,2}^{(0)}&=\E_{\lambda_0}^{(n)}\left[\left|\int_0^T\log^2\left(\frac{\lambda_0(t)}{\lambda(t)}\right)[\d N_t-Y_t\lambda_0(t)\d t]\right|^k\right],\\
B_{k,2}^{(1)}&=\E_{\lambda_0}^{(n)}\left[\left|\int_0^T\log^2\left(\frac{\lambda_0(t)}{\lambda(t)}\right)[Y_t-\mu_n(t)]\lambda_0(t)\d t\right|^k\right],\\
B_{k,2}^{(2)}&=\left|\int_0^T\log^2\left(\frac{\lambda_0(t)}{\lambda(t)}\right)\mu_n(t)\lambda_0(t)\d t\right|^k.
\end{split}
\end{eqnarray*}
This can be iterated: we set $J=\min\{j\in\mathbb{N}: \ 2^j\geq k\}$ so that  $1<k2^{1-J}\leq 2$. There exists a constant $C_k$, only depending on $k$, such that for
$$B_{k2^{1-j},2^j}^{(1)}=\E_{\lambda_0}^{(n)}\left[\left|\int_0^T\log^{2^j}\left(\frac{\lambda_0(t)}{\lambda(t)}\right)[Y_t-\mu_n(t)]\lambda_0(t)\d t\right|^{k2^{1-j}}\right]$$
and
$$B_{k2^{1-j},2^j}^{(2)}=\left|\int_0^T\log^{2^j}\left(\frac{\lambda_0(t)}{\lambda(t)}\right)\mu_n(t)\lambda_0(t)\d t\right|^{k2^{1-j}},$$
\begin{eqnarray*}
\begin{split}
B_{2k}&\leq C_k\left(\E_{\lambda_0}^{(n)}\left[\left|\int_0^T\log^{2^J}\left(\frac{\lambda_0(t)}{\lambda(t)}\right)[\d N_t-Y_t\lambda_0(t) \d t]\right|^{k2^{1-J}}\right]\right.\\
&\left.\hspace*{8cm}+\sum_{j=1}^J(B_{k2^{1-j},2^j}^{(1)}+B_{k2^{1-j},2^j}^{(2)}) \right)\\
&\leq C_k\left\{\left(\E_{\lambda_0}^{(n)}\left[\left|\int_0^T\log^{2^J}\left(\frac{\lambda_0(t)}{\lambda(t)}\right)[\d N_t-Y_t\lambda_0(t)\d t]\right|^2\right]\right)^{k2^{-J}}\right.\\
&\left.\hspace*{8cm}+\sum_{j=1}^J(B_{k2^{1-j},2^j}^{(1)}+B_{k2^{1-j},2^j}^{(2)}) \right\}\\
&= C_k\left[\left(\E_{\lambda_0}^{(n)}\left[\int_0^T\log^{2^{J+1}}\left(\frac{\lambda_0(t)}{\lambda(t)}\right)Y_t\lambda_0(t)\d t\right]\right)^{k2^{-J}}+\sum_{j=1}^J(B_{k2^{1-j},2^j}^{(1)}+B_{k2^{1-j},2^j}^{(2)}) \right]\\
&= C_k\left[\left(\int_0^T\log^{2^{J+1}}\left(\frac{\lambda_0(t)}{\lambda(t)}\right)\mu_n(t)\lambda_0(t)\d t\right)^{k2^{-J}}+\sum_{j=1}^J(B_{k2^{1-j},2^j}^{(1)}+B_{k2^{1-j},2^j}^{(2)}) \right]\\
&= C_k\left[B_{k2^{-J},2^{J+1}}^{(2)}+\sum_{j=1}^J(B_{k2^{1-j},2^j}^{(1)}+B_{k2^{1-j},2^j}^{(2)}) \right].
\end{split}
\end{eqnarray*}
Note that, for any $1\leq j\leq J$,
\begin{eqnarray*}
\begin{split}
B_{k2^{1-j},2^j}^{(1)}&\leq \left[\int_0^T\log^{2^{j+1}}\left(\frac{\lambda_0(t)}{\lambda(t)}\right)\lambda_0^2(t)\d t\right]^{k2^{-j}}\times \E_{\lambda_0}^{(n)}\left[\left(\int_0^T[Y_t-\mu_n(t)]^2\d t\right)^{k2^{-j}}\right]\\
&\leq C(M_{\lambda_0}^2\|\bar\lambda_0\|_\infty)^{k2^{-j}}\left[\int_0^T\log^{2^{j+1}}\left(\frac{M_{\lambda_0}\bar\lambda_0(t)}{M_{\lambda}\bar\lambda(t)}\right)\bar\lambda_0(t)\d t\right]^{k2^{-j}}\times n^{k2^{-j}}\\
&\leq C\left[\log^{2^{j+1}}\left(\frac{M_{\lambda_0}}{M_\lambda}\right)+E_{2^{j+1}}(\bar\lambda_0;\,\bar\lambda)\right]^{k2^{-j}}\times n^{k2^{-j}}\\
&\leq C(nv_n^2)^{k2^{-j}}\leq C(nv_n^2)^k,
\end{split}
\end{eqnarray*}
where we have used \eqref{moment}. Similarly,  for any $j\geq 1$,
\begin{eqnarray*}
\begin{split}
B_{k2^{1-j},2^j}^{(2)}&\leq (nm_2M_{\lambda_0})^{k2^{1-j}}
\left[\int_0^T\log^{2^j}\left(\frac{M_{\lambda_0}\bar\lambda_0(t)}{M_{\lambda}\bar\lambda(t)}\right)\bar\lambda_0(t)\d t\right]^{k2^{1-j}}\\
&\leq C\left[\log^{2^j}\left(\frac{M_{\lambda_0}}{M_\lambda}\right)+E_{2^j}(\bar\lambda_0;\,\bar\lambda)\right]^{k2^{1-j}} \times n^{k2^{1-j}}
\leq C(nv_n^2)^{k2^{1-j}} \leq C(nv_n^2)^k.
\end{split}
\end{eqnarray*}
Therefore, for any $k>1$,
$$V_{2k}(\lambda_0;\, \lambda)\leq \kappa(nv_n^2)^k,$$
where $\kappa$ depends on $C_{1k}$, $k$, $H$, $\lambda_0$, $m_1$ and $m_2$. Using previous computations, the case $k=1$ is straightforward.
So, we obtain the result for $V_{k}(\lambda_0;\, \lambda)$ for every $k\geq 2$.
\end{proof}
%%%%%%%%%%%%%%%%%%%%%%%%%%%%%%%

To prove Proposition \ref{prop:test:Aalen}, we use the following lemma whose proof is reported in Section \ref{app}.

%%%%%%%%%%%%%%%%%%%%%%%%%%%%%%%%%%%%%%%%%%%%%%%%%%%%%%%%%%%%%%%%%%%%%%%%%%%%%%%%%%%%%%%%%%%%%%%%%%%%%%%%%%%%%%%%%%%%%%%%%%%%
\begin{lem} \label{lem:tests}
Under condition \eqref{ass:Y1}, there exist constants $\xi,\,K>0$, only depending on $M_{\lambda_0}$, $\alpha,$ $m_1$ and $m_2$, such that, for any non-negative function $\lambda_1$,
there exists a test $\phi_{\lambda_1}$ so that
$$\E_{\lambda_0}^{(n)}[\1_{\Gamma_n}\phi_{\lambda_1}]\leq 2\exp\left(- K n\| \lambda_1 - \lambda_0 \|_{1} \times\min\{\| \lambda_1 - \lambda_0 \|_1,\,m_1\} \right)$$
and
$$\sup_{\lambda: \ \| \lambda - \lambda_1 \|_{1} <\xi\| \lambda_1 - \lambda_0 \|_1 } \E_{\lambda}[\1_{\Gamma_n}(1-\phi_{\lambda_1})]\leq 2\exp\left(-Kn\| \lambda_1 - \lambda_0 \|_1\times\min\{\| \lambda_1 - \lambda_0 \|_{1},\,m_1\} \right).$$
\end{lem}

%\smallskip

\begin{proof}[Proof of Proposition~\ref{prop:test:Aalen}]
We consider the setting of Lemma~\ref{lem:tests} and a covering of $S_{n,j}(v_n)$ with $\L_1$-balls of radius $\xi j v_n$ and centers $(\lambda_{l,j})_{l=1,\,\ldots,\, D_j }$, where $D_j$ is the covering number of $S_{n,j}(v_n)$ by such balls. We set $\phi_{n,j} = \max_{l=1,\,\ldots,\, D_j}\phi_{\lambda_{l,j}}$, where the $\phi_{\lambda_{l,j}}$'s are defined in Lemma~\ref{lem:tests}. So, there exists a constant $\rho>0$ such that
 $$ \E_{\lambda_0}^{(n)}[ \1_{\Gamma_n} \phi_{n,j} ] \leq 2D_j e^{-K n j^2v_n^2} \,\,\,\mbox{and} \,\,  \sup_{\lambda \in S_{n,j}(v_n)}
 \E_{\lambda}^{(n)}[ \1_{\Gamma_n} ( 1 -\phi_{n,j}) ] \leq 2e^{-K n j^2v_n^2 },\quad\mbox{if }  j \leq \frac{\rho}{v_n},$$
and
 $$ \E_{\lambda_0}^{(n)}[ \1_{\Gamma_n} \phi_{n,j} ] \leq 2D_j e^{-K n jv_n } \,\,\,\mbox{and} \,\,  \sup_{\lambda \in S_{n,j}(v_n)} \E_{\lambda}^{(n)}[ \1_{\Gamma_n} ( 1 -\phi_{n,j}) ] \leq 2e^{-K n jv_n},\quad\mbox{if }  j > \frac{\rho}{v_n},$$
where $K$ is a constant (see Lemma~\ref{lem:tests}). We now bound $D_j$.
First note that for any $\lambda=M_\lambda\bar \lambda$ and $\lambda' = M_{\lambda'}\bar \lambda'$,
\begin{equation}\label{llbar}
\|\lambda -\lambda'\|_1 \leq M_\lambda\|\bar \lambda - \bar \lambda'\|_1+ |M_\lambda-M_{\lambda'}|.
\end{equation}
Assume that $M_\lambda \geq M_{\lambda_0}$. Then,
\begin{equation*}
\begin{split}
 \|\lambda - \lambda_0\|_1 & \geq \int_{\bar \lambda >\bar \lambda_0} [M_\lambda \bar \lambda (t) - M_{\lambda_0} \bar \lambda_0 (t)]\d t \\
 &= M_\lambda\int_{\bar \lambda >\bar \lambda_0} [\bar \lambda (t) - \bar \lambda_0 (t)]\d t + (M_\lambda - M_{\lambda_0})
 \int_{\bar \lambda >\bar \lambda_0}\bar \lambda_0 (t)\d t\\
 &\geq  M_\lambda\int_{\bar \lambda >\bar \lambda_0} [\bar \lambda (t) - \bar \lambda_0 (t)] \d t = \frac{ M_\lambda}{ 2}\|\bar \lambda -\bar \lambda_0\|_1.
 \end{split}
 \end{equation*}
Conversely, if $M_\lambda < M_{\lambda_0}$,
\begin{equation*}
\begin{split}
 \|\lambda - \lambda_0\|_1 & \geq \int_{\bar \lambda_0 >\bar \lambda} [M_{\lambda_0} \bar \lambda_0 (t) - M_\lambda \bar \lambda (t)]\d t \\
 &\geq M_{\lambda_0}\int_{\bar \lambda_0 >\bar \lambda} [\bar \lambda_0 (t) -  \bar \lambda (t)] \d t =
 \frac{ M_{\lambda_0}}{ 2 }\|\bar \lambda -\bar \lambda_0\|_1.
 \end{split}
 \end{equation*}
 So,
$2  \| \lambda - \lambda_0\|_1 \geq  (M_\lambda \vee M_{\lambda_0})  \|\bar \lambda - \bar \lambda_0\|_1$
 and  we finally have
 \begin{equation} \label{norm:mino}
 \| \lambda - \lambda_0\|_1 \geq \max \left\{{(M_\lambda \vee M_{\lambda_0})  \|\bar \lambda - \bar \lambda_0\|_1 }/{ 2 } , \,| M_\lambda - M_{\lambda_0}|  \right\}.
 \end{equation}
So, for all $\lambda = M_\lambda \bar \lambda  \in S_{n,j}(v_n)$,
\begin{equation}\label{7.9}
\| \bar \lambda -\bar \lambda_0 \|_1 \leq \frac{2(j+1) v_n}{M_{\lambda_0}} \quad\mbox{and}\quad |M_\lambda-M_{\lambda_0}| \leq (j+1) v_n.
 \end{equation}
 Therefore, $S_{n,j}(v_n)\subseteq  (\bar S_{n,j} \cap \mathcal F_n)\times \{M:\, \ |M-M_{\lambda_0}| \leq  (j+1) v_n\}$ and any covering of $(\bar S_{n,j} \cap \mathcal F_n)\times \{M:\, \ |M-M_{\lambda_0}| \leq  (j+1) v_n\}$ will give a covering of $S_{n,j}(v_n)$. So, to bound $D_j$, we have to build a convenient covering of $( \bar S_{n,j} \cap \mathcal F_n)\times \{M:\, \ |M-M_{\lambda_0}| \leq  (j+1) v_n\}$. We distinguish two cases.
\begin{itemize}
\item We assume that $(j+1)v_n \leq 2M_{\lambda_0}$. Then, \eqref{7.9} implies that $M_\lambda \leq 3 M_{\lambda_0}$.
 Moreover, if
 $$\|\bar \lambda -\bar  \lambda'\|_1 \leq \frac{\xi j v_n}{3M_{\lambda_0}+1} \quad \mbox{ and } \quad |M_\lambda-M_{\lambda'}|\leq \frac{\xi j v_n}{3M_{\lambda_0}+1},$$
 then, by \eqref{llbar},
 $$\| \lambda - \lambda'\|_1\leq \frac{(M_\lambda +1) \xi j v_n}{3M_{\lambda_0}+1} \leq \xi j v_n.$$
By assumption $(ii)$ of Theorem \ref{th:gene:aalen}, this implies that, for any $\delta>0$, there exists $J_0$ such that for $j\geq J_0$,
  \begin{eqnarray*}
  \begin{split}
  D_j&\leq D( (3 M_{\lambda_0}+1)^{-1}  \xi j v_n,\, \bar S_{n,j} \cap \mathcal F_n,\, \| \cdot \|_1 ) \times \left[2(j+1)v_n\times\frac{(3 M_{\lambda_0}+1) }{ \xi jv_n}+\frac{1}{2}\right]\\
  &\lesssim \exp( \delta (j+1)^2 nv_n^2).
  \end{split}
  \end{eqnarray*}
\item We assume that $(j+1) v_n > 2 M_{\lambda_0}$.   If $$ \| \bar \lambda - \bar \lambda' \|_1 \leq \frac{\xi}{4} \quad  \mbox{ and } \quad |M_\lambda-M_{\lambda'}| \leq \frac{\xi (M_\lambda\vee M_{\lambda_0})}{ 4}, $$
   using again \eqref{llbar} and \eqref{7.9},
  $$
  \| \lambda - \lambda'\|_1\leq \frac{\xi M_\lambda }{4}+\frac{\xi (M_\lambda+M_{\lambda_0})}{4}
  \leq\frac{3\xi M_{\lambda_0}}{4}+\frac{\xi (j+1)v_n}{2}\leq \frac{7\xi (j+1)v_n}{8}\leq \xi jv_n,$$
for $n$ large enough. By assumption $(i)$ of Theorem \ref{th:gene:aalen}, this implies that, for any $\delta>0$,
  $$D_j\lesssim D( \xi /4,\, \mathcal F_n,\, \| \cdot \|_1 ) \times \log((j+1)v_n)\lesssim \log(jv_n)\exp (\delta n).$$
 \end{itemize}
It is enough to choose $\delta$ small enough to obtain the result of Proposition \ref{prop:test:Aalen}.
\end{proof}
%%%%%%%%%%%%%%%%%%%%%%%%%%%%%%%
\section[Appendix]{Appendix}\label{app}

\begin{proof}[Proof of Lemma \ref{lem:tests}]
For any $\lambda$, we denote by $\E^{(n)}_{\lambda,\Gamma_n}[\cdot]=\E^{(n)}_{\lambda}[\1_{\Gamma_n}\times\cdot]$. For any $\lambda,\,\lambda'$,
we define
$$ \|\lambda - \lambda'\|_{\tilde \mu_n}:=\int_\Omega |\lambda(t)-\lambda'(t)|\tilde \mu_n(t)\d t.$$
On $\Gamma_n$ we have
\begin{equation}\label{comp:norm}
m_1 \| \lambda - \lambda_0\|_1 \leq \| \lambda - \lambda_0\|_{\tilde \mu_n} \leq m_2 \| \lambda - \lambda_0\|_1.
\end{equation}
The main tool for building convenient tests is Theorem~3 of \cite{HRR} (and  its proof) applied in the univariate setting. By mimicking the proof of this theorem from Inequality (7.5) to Inequality (7.7), if $H$ is a deterministic function bounded by $b$, we have that, for any $u\geq 0$,
\begin{equation}\label{concentration}
\P_\lambda^{(n)}\left(\left|\int_0^T H_t(\d N_t-\d \Lambda_t)\right|\geq \sqrt{2v u } + \frac{bu}{3} \mbox{ and } \Gamma_n\right)\leq 2e^{-u},
\end{equation}
where we recall that  $ \Lambda_t = \int_0^t Y_s\lambda(s)\d s$ and $v $ is a deterministic constant such that, on $\Gamma_n$,
$\int_0^T H_t^2 Y_t\lambda(t)\d t \leq  v $ almost surely.
For any non-negative function $\lambda_1$, we define the sets
$$A:=\{t\in\Omega:\ \lambda_1(t)\geq \lambda_0(t)\}\quad \mbox{and} \quad A^c:=\{t\in\Omega:\ \lambda_1(t)< \lambda_0(t)\}$$ and the following pseudo-metrics
$$d_A(\lambda_1,\,\lambda_0):=\int_A[\lambda_1(t)-\lambda_0(t)]\tilde\mu_n(t)\d t \quad \mbox{and} \quad d_{A^c}(\lambda_1,\,\lambda_0):=\int_{A^c}[\lambda_0(t)-\lambda_1(t)]\tilde\mu_n(t)\d t.$$
Note that
$ \| \lambda_1 - \lambda_0 \|_{\tilde\mu_n}  = d_A(\lambda_1,\,\lambda_0) +  d_{A^c}(\lambda_1,\,\lambda_0).$
For $u>0$, if $d_A(\lambda_1, \lambda_0) \geq d_{A^c}(\lambda_1, \lambda_0)$, define the test
\[\phi_{\lambda_1,A}(u):=\1\left\{N(A)-\int_{A} \lambda_0(t)Y_t\d t \geq  \rho_n(u) \right\} ,\quad \mbox{ with  } \rho_n(u) := \sqrt{2n v(\lambda_0) u} + \frac{u}{3},\]
where, for any non-negative function $\lambda$,
\begin{equation}\label{vl}
v(\lambda) := (1+\alpha)\int_\Omega \lambda(t) \tilde\mu_n(t)\d t.
\end{equation}
Similarly, if
$d_A(\lambda_1,\, \lambda_0) < d_{A^c}(\lambda_1,\, \lambda_0)$, define
\[\phi_{\lambda_1,A^c}(u):=\1\left\{N(A^c)-\int_{A^c}\lambda_0(t)Y_t\d t \leq - \rho_n(u) \right\}. \]
Since for any non-negative function $\lambda$,  on $\Gamma_n$, by \eqref{ass:Y3},
\begin{equation}
\label{Yn:mun}
(1-\alpha) \int_\Omega \lambda(t) \tilde\mu_n(t)\d t\leq  \int_\Omega \lambda(t)\frac{Y_t}{n} \d t \leq (1+\alpha)\int_\Omega\lambda(t) \tilde\mu_n(t)\d t,
\end{equation}
inequality \eqref{concentration} applied with $H=\1_A$ or $H=\1_{A^c}$,  $b=1$ and $v=n v(\lambda_0)$ implies that, for any $u>0$,
\begin{equation}\label{err1}
\E^{(n)}_{\lambda_0,\Gamma_n}[\phi_{\lambda_1,A}(u) ]\leq 2e^{-u }\quad\mbox{and}\quad \E^{(n)}_{\lambda_0,\Gamma_n}[\phi_{\lambda_1,A^c}(u) ]\leq 2e^{-u }.
\end{equation}

%\medskip

We now state a useful lemma whose proof is given below.

\medskip

\begin{lem} \label{lem:tests:aalen:error2}
Assume condition \eqref{ass:Y1} is verified.
Let $\lambda$ be a non-negative function. Assume that $$\| \lambda - \lambda_1\|_{\tilde\mu_n} \leq \frac{1-\alpha}{ 4(1+\alpha)} \| \lambda_1 -\lambda_0\|_{\tilde\mu_n}.$$
We set $\tilde M_n(\lambda_0)=\int_\Omega\lambda_0(t)\tilde\mu_n(t)\d t$
 and we distinguish two cases.
\begin{enumerate}
\item Assume that $d_A(\lambda_1,\, \lambda_0) \geq d_{A^c}(\lambda_1,\, \lambda_0)$. Then,
$$\E^{(n)}_{\lambda,\Gamma_n}[1-\phi_{\lambda_1,A}(u_A)]\leq 2\exp(-u_A),$$
where
$$
u_A=\left\{\begin{array}{lc}
u_{0A}n d_A^2(\lambda_1,\,\lambda_0),&\mbox{ if }  \| \lambda_1 -\lambda_0\|_{\tilde\mu_n}\leq 2\tilde M_n(\lambda_0),\\[2pt]
u_{1A}nd_A(\lambda_1,\lambda_0),&\mbox{ if }  \| \lambda_1 -\lambda_0\|_{\tilde\mu_n}> 2\tilde M_n(\lambda_0),
\end{array}\right.
$$
and $u_{0A}$, $u_{1A}$ are two constants only depending on $\alpha$, $M_{\lambda_0}$, $m_1$ and $m_2$.
\item Assume that $d_A(\lambda_1,\, \lambda_0) < d_{A^c}(\lambda_1,\, \lambda_0)$. Then,
$$\E^{(n)}_{\lambda,\Gamma_n}[1-\phi_{\lambda_1,A^c}(u_{A^c})]\leq 2\exp(-u_{A^c}),$$
where
$$
u_{A^c}=\left\{\begin{array}{lc}
u_{0A^c}nd_{A^c}^2(\lambda_1,\,\lambda_0),&\mbox{ if }  \| \lambda_1 -\lambda_0\|_{\tilde\mu_n}\leq 2\tilde M_n(\lambda_0),\\[2pt]
u_{1A^c}nd_{A^c}(\lambda_1,\,\lambda_0),&\mbox{ if }  \| \lambda_1 -\lambda_0\|_{\tilde\mu_n}> 2\tilde M_n(\lambda_0),
\end{array}\right.
$$
and $u_{0A^c}$, $u_{1A^c}$ are two constants only depending on $\alpha$, $M_{\lambda_0}$, $m_1$ and $m_2$.
\end{enumerate}
\end{lem}

%%%%%%%%%%%%%%%%%%%%%%%%%%%%%%%%%%%%%%%%%%%%%%%%%%%%%%%%%%%%%%%%%%%%%%%%%%%%%%%%%%%%%%%%%%%%%%%%%%%%%%%%%%%%%%%%%%%%%%%%%%%%%%%%%%%%%%%%%

\bigskip

%%%%%%%%%%%%%%%%%%%%%%%%%%%%%%%%%%%%%%%%%%%%%%%%%%%%%%%%%%%%%%%%%%%%%%%%%%%%%%%%%%%%%%%%%%%%%%%%%%%%%%%%%%%%%%%%%%%%%%%%%%%%%%%%%%%%%%%%%%%%%%%%%%%

%%%%%%%%%%%%%%%%%%%%%%%%%%%%%%%%%%%%%%%%%%%%%%%%%%%%%%%%%%%%%%%%%%%%%%%%%%%%%%%%%%%%%%%%%%%%%%%%%%%%%%%%%%%%%%%%%%%%%%%%%%%%%%%%%%%%%%%%%%%%%%%%%%%%%%%

Note that, by \eqref{comp:norm}, if $d_A(\lambda_1,\, \lambda_0)\geq d_{A^c}(\lambda_1,\, \lambda_0)$, by virtue of Lemma \ref{lem:tests:aalen:error2},
\begin{eqnarray*}
\begin{split}
u_A&\geq\min\{u_{0A}n d_A^2(\lambda_1,\,\lambda_0),\,u_{1A}nd_A(\lambda_1,\,\lambda_0)\}\\
&\geq nd_A(\lambda_1,\,\lambda_0)\times\min\{u_{0A} d_A(\lambda_1,\,\lambda_0),\,u_{1A}\}\\
&\geq \frac{1}{2}nm_1\|\lambda_1-\lambda_0\|_1\times\min\left\{\frac{1}{2}u_{0A}m_1\|\lambda_1-\lambda_0\|_1,\,u_{1A}\right\}\\
&\geq  K_A n\| \lambda_1 - \lambda_0 \|_{1} \times\min\{\| \lambda_1 - \lambda_0 \|_1,\,m_1\},
\end{split}
\end{eqnarray*}
for $K_A$ a positive constant small enough only depending on $\alpha$, $M_{\lambda_0}$, $m_1$ and $m_2$.
Similarly, if $d_A(\lambda_1,\, \lambda_0) < d_{A^c}(\lambda_1,\, \lambda_0)$,
\begin{eqnarray*}
\begin{split}
u_{A^c}&\geq \frac{1}{2}nm_1\|\lambda_1-\lambda_0\|_1\times\min\left\{\frac{1}{2}u_{0A^c}m_1\|\lambda_1-\lambda_0\|_1,\,u_{1A^c}\right\}\\
&\geq K_{A^c} n\| \lambda_1 - \lambda_0 \|_{1} \times\min\{\| \lambda_1 - \lambda_0 \|_1,\,m_1\},
\end{split}
\end{eqnarray*}
for $K_{A^c}$  a positive constant small enough only depending on $\alpha$, $M_{\lambda_0}$, $m_1$ and $m_2$.
Now, we set
$$\phi_{\lambda_1}=\phi_{\lambda_1,A}(u_A)\1_{\left\{d_A(\lambda_1,\, \lambda_0)\geq d_{A^c}(\lambda_1,\, \lambda_0)\right\}}+\phi_{\lambda_1,A^c}(u_{A^c})\1_{\left\{d_A(\lambda_1,\, \lambda_0)< d_{A^c}(\lambda_1,\, \lambda_0)\right\}},$$
so that, with
$K=\min\{K_A,\,K_{A^c}\}$, by using \eqref{err1},
\begin{eqnarray*}
\begin{split}
\E^{(n)}_{\lambda_0,\Gamma_n}[\phi_{\lambda_1}]&= \E^{(n)}_{\lambda_0,\Gamma_n}[\phi_{\lambda_1,A}(u_A)]\1_{\left\{d_A(\lambda_1,\, \lambda_0)\geq d_{A^c}(\lambda_1,\, \lambda_0)\right\}}\\&\hspace*{4cm} +\E^{(n)}_{\lambda_0,\Gamma_n}[\phi_{\lambda_1,A^c}(u_{A^c})]\1_{\left\{d_A(\lambda_1,\, \lambda_0)< d_{A^c}(\lambda_1,\, \lambda_0)\right\}}\\
&\leq 2e^{-u_A}1_{\left\{d_A(\lambda_1,\, \lambda_0)\geq d_{A^c}(\lambda_1,\, \lambda_0)\right\}}+2e^{-u_{A^c}}1_{\left\{d_A(\lambda_1, \,\lambda_0)< d_{A^c}(\lambda_1,\, \lambda_0)\right\}}\\
&\leq 2\exp\left(- K n\| \lambda_1 - \lambda_0 \|_{1} \times\min\{\| \lambda_1 - \lambda_0 \|_1,\,m_1\} \right).
\end{split}
\end{eqnarray*}
If $\|\lambda - \lambda_1 \|_{1} <\xi\| \lambda_1 - \lambda_0 \|_1$, $ \xi=m_1(1-\alpha)/[4m_2(1+\alpha)]$,
then
$$\| \lambda - \lambda_1\|_{\tilde\mu_n} \leq \frac{1-\alpha}{ 4(1+\alpha)} \| \lambda_1 -\lambda_0\|_{\tilde\mu_n}$$
and Lemma \ref{lem:tests:aalen:error2} shows that
\begin{eqnarray*}
\begin{split}
\E^{(n)}_{\lambda,\Gamma_n}[1-\phi_{\lambda_1}]&\leq 2e^{-u_A}\1_{\left\{d_A(\lambda_1,\, \lambda_0)\geq d_{A^c}(\lambda_1, \, \lambda_0)\right\}}+2e^{-u_{A^c}}\1_{\left\{d_A(\lambda_1,\, \lambda_0)< d_{A^c}(\lambda_1,\, \lambda_0)\right\}}\\
&\leq 2\exp\left(- K n\| \lambda_1 - \lambda_0 \|_{1} \times\min\{\| \lambda_1 - \lambda_0 \|_1,\,m_1\} \right),
\end{split}
\end{eqnarray*}
which completes the proof of Lemma \ref{lem:tests}.
\end{proof}
%%%%%%%%%%%%%%%%%%%%%%%%%%%%

\begin{proof}[Proof of Lemma \ref{lem:tests:aalen:error2}]
We only consider the case where $d_A(\lambda_1,\, \lambda_0)\geq d_{A^c}(\lambda_1,\, \lambda_0)$.
The case $d_A(\lambda_1,\, \lambda_0)< d_{A^c}(\lambda_1,\, \lambda_0)$ can be dealt with using similar arguments.
So, we assume that $d_A(\lambda_1,\, \lambda_0) \geq d_{A^c}(\lambda_1, \,\lambda_0)$. On $\Gamma_n$ we have
\begin{eqnarray*}
\begin{split}
\int_A[\lambda_1(t)-\lambda_0(t)]Y_t\d t&\geq n(1-\alpha)\int_A[\lambda_1(t)-\lambda_0(t)]\tilde\mu_n(t)\d t\\
&\geq \frac{n(1-\alpha)}{2}\|\lambda_1-\lambda_0\|_{\tilde\mu_n}\\ &\geq2n(1+\alpha)\|\lambda-\lambda_1\|_{\tilde\mu_n}\\
&\geq 2n(1+\alpha)\int_A|\lambda(t)-\lambda_1(t)|\tilde\mu_n(t)\d t
\geq 2\int_A|\lambda(t)-\lambda_1(t)|Y_t\d t.
\end{split}
\end{eqnarray*}
Therefore,
\begin{eqnarray*}
\begin{split}
\E^{(n)}_{\lambda,\Gamma_n}[1- \phi_{\lambda_1,A}(u_A)] &= \P^{(n)}_{\lambda,\Gamma_n} \left(N(A)- \int_{A} \lambda(t)Y_t\d t <  \rho_n(u_A)+ \int_{A} (\lambda_0-\lambda)(t)Y_t\d t  \right)  \\
&= \P^{(n)}_{\lambda,\Gamma_n}  \left( N(A)- \int_{A} \lambda(t)Y_t\d t <  \rho_n(u_A)- \int_{A} (\lambda_1-\lambda_0)(t)Y_t\d t\right. \\
&\hspace{6.7cm}+\left.\int_{A} (\lambda_1-\lambda)(t)Y_t\d t\right) \\
&\leq \P^{(n)}_{\lambda,\Gamma_n}  \left( N(A)- \int_{A} \lambda(t)Y_t\d t <  \rho_n(u_A)- \frac{1}{2}\int_{A} (\lambda_1-\lambda_0)(t)Y_t\d t\right).
\end{split}
\end{eqnarray*}

Assume that $\|\lambda_1-\lambda_0\|_{\tilde\mu_n}\leq 2\tilde M_n(\lambda_0)$. This assumption implies that
$d_A(\lambda_1,\,\lambda_0)\leq\|\lambda_1-\lambda_0\|_{\tilde\mu_n}\leq 2\tilde M_n(\lambda_0)\leq 2m_2M_{\lambda_0}$.
Since
$v(\lambda_0)= (1+\alpha)\tilde M_n(\lambda_0),$ with
$u_A=u_{0A}n d_A^2(\lambda_1,\, \lambda_0),$
where $u_{0A}\leq 1$ is a constant depending on $\alpha$, $m_1$ and $m_2$ chosen later, we  have
$$ \rho_n(u_A) \leq n d_A(\lambda_1,\, \lambda_0)\sqrt{2u_{0A}(1+\alpha)\tilde M_n(\lambda_0)}+ \frac{ u_{0A} n d_A^2(\lambda_1,\, \lambda_0)}{ 3 }  \leq K_1\sqrt{u_{0A}}  n d_A(\lambda_1, \,\lambda_0) $$
as soon as $K_1\geq[2(1+\alpha)\tilde M_n(\lambda_0)]^{1/2}+ {2\tilde M_n(\lambda_0)\sqrt{ u_{0A}}}/{3}$. Note that
the definition of $v(\lambda)$ in \eqref{vl} gives
\begin{eqnarray*}
\begin{split}
v(\lambda) &= (1+\alpha)\int_\Omega\lambda_0(t)\tilde\mu_n(t)\d t+ (1+\alpha)\int_\Omega[\lambda(t)-\lambda_0(t)]\tilde\mu_n(t)\d t\\
&\leq v(\lambda_0)+ (1+\alpha)\|\lambda -\lambda_0\|_{\tilde\mu_n}\\
&\leq  v(\lambda_0) + (1+\alpha) \left[ \|\lambda -\lambda_1\|_{\tilde\mu_n} + \|\lambda_1 -\lambda_0\|_{\tilde\mu_n}\right] \\
&\leq v(\lambda_0) + \frac{5+3\alpha}{4}\|\lambda_1 -\lambda_0\|_{\tilde\mu_n} \leq C_1,
\end{split}
\end{eqnarray*}
where $C_1$ only depends on $\alpha,$ $M_{\lambda_0},$ $m_1$ and $m_2$.
Combined with \eqref{Yn:mun},  this implies that, on $\Gamma_n$, if $K_1\leq {(1-\alpha)}/[4\sqrt{u_{0A}}]$, which is true for $u_{0A}$ small enough,
\begin{equation*}
\begin{split}
 \frac{1}{2}\int_{A} (\lambda_1-\lambda_0)(t)Y_t\d t-  \rho_n(u_A)
&\geq \frac{(1-\alpha)n}{ 2 } d_A(\lambda_1,\, \lambda_0) \left[ 1 - \frac{2K_1\sqrt{u_{0A}}}{1-\alpha}\right] \\
&\geq
 \frac{ (1-\alpha)n}{4 } d_A(\lambda_1,\, \lambda_0) \geq \sqrt{ 2 nC_1r} + \frac{r}{3}\geq \sqrt{ 2 nv(\lambda) r} + \frac{r}{3},
\end{split}
\end{equation*}
with
$$ r = n\min \left\{
\frac{(1-\alpha)^2  }{ 128 C_1 }d_A^2(\lambda_1,\, \lambda_0)  , \,\frac{ 3 (1-\alpha) }{ 8}d_A(\lambda_1,\, \lambda_0) \right\}.$$
Inequality \eqref{concentration} then leads to
\begin{equation}\label{r}
\E^{(n)}_{\lambda,\Gamma_n}[1- \phi_{\lambda_1,A}(u_A) ] \leq 2e^{-r}.
\end{equation}
For $u_{0A}$ small enough only depending on $M_{\lambda_0}$, $\alpha$, $m_1$ and $m_2$, we have
$$\frac{(1-\alpha)}{4\sqrt{u_{0A}}}\geq \sqrt{2(1+\alpha)\tilde M_n(\lambda_0)}+ \frac{2\tilde M_n(\lambda_0)\sqrt{u_{0A}}}{3}$$
so (\ref{r}) is true. Since $r\geq u_A$ for $u_{0A}$ small enough, then
$$
\E^{(n)}_{\lambda,\Gamma_n}[1- \phi_{\lambda_1,A}(u) ] \leq 2e^{-u_A}.
$$

Assume that $\|\lambda_1-\lambda_0\|_{\tilde\mu_n}> 2\tilde M_n(\lambda_0)$. We take
$u_A=u_{1A}nd_A(\lambda_1,\,\lambda_0),$
where $u_{1A}\leq 1$ is a constant depending on $\alpha$ chosen later. We still consider the same test $\phi_{\lambda_1, A}(u_A)$. Observe now that, since $d_A(\lambda_1,\,\lambda_0)\geq \frac{1}{2}\|\lambda_1 -\lambda_0\|_{\tilde\mu_n}\geq \tilde M_n(\lambda_0)$,
\begin{eqnarray*}
\begin{split}
\rho_n(u_A)&= \sqrt{2nu_Av(\lambda_0)}+\frac{u_A}{3}\\
&\leq n\sqrt{2(1+\alpha)u_{1A}\tilde M_n(\lambda_0)d_A(\lambda_1,\,\lambda_0)}+\frac{nu_{1A}}{3}d_A(\lambda_1,\,\lambda_0)\\
&\leq \left[\sqrt{2(1+\alpha)}+\frac{1}{3}\right]n\sqrt{u_{1A}}d_A(\lambda_1,\,\lambda_0)
\end{split}
\end{eqnarray*}
and, under the assumptions of the lemma,
\begin{equation}\label{v}
v(\lambda)\leq(1+\alpha)\tilde M_n(\lambda_0) + (1+\alpha) \left[\|\lambda -\lambda_1\|_{\tilde\mu_n} + \|\lambda_1 -\lambda_0\|_{\tilde\mu_n}\right]\leq C_2d_A(\lambda_1,\,\lambda_0),
\end{equation}
where $C_2$ only depends on $\alpha.$ Therefore,
\begin{eqnarray*}\begin{split}
&\frac{1}{2}\int_{A} (\lambda_1-\lambda_0)(t)Y_t\d t-\rho_n(u_A)\\&\hspace*{1.5cm}\geq \frac{n(1-\alpha)}{2}\int_A[\lambda_1(t)-\lambda_0(t)]\tilde\mu_n(t)\d t -\left(\sqrt{2(1+\alpha)}+\frac{1}{3}\right)\sqrt{u_{1A}}nd_A(\lambda_1,\,\lambda_0)\\
&\hspace*{1.5cm}\geq \left[\frac{1-\alpha}{2}-\left(\sqrt{2(1+\alpha)}+\frac{1}{3}\right)\sqrt{u_{1A}}\right]nd_A(\lambda_1,\,\lambda_0)\\
&\hspace*{1.5cm}\geq \frac{1-\alpha}{4}nd_A(\lambda_1,\,\lambda_0),
\end{split}
\end{eqnarray*}
where the last inequality is true for $u_{1A}$ small enough depending only on $\alpha$. Finally, using (\ref{v}), since $u_A=u_{1A}nd_A(\lambda_1,\,\lambda_0)$, we have
\begin{eqnarray*}
\begin{split}
\frac{1-\alpha}{4}nd_A(\lambda_1,\,\lambda_0)&\geq\sqrt{2nC_2d_A(\lambda_1,\,\lambda_0)u_{1A}nd_A(\lambda_1,\,\lambda_0)}+\frac{1}{3}u_{1A}nd_A(\lambda_1,\,\lambda_0)\\
&\geq \sqrt{2nv(\lambda)u_A}+\frac{u_A}{3}
\end{split}
\end{eqnarray*}
for $u_{1A}$ small enough depending only on $\alpha$. We then obtain
$$\E^{(n)}_{\lambda,\Gamma_n}[1- \phi_{\lambda_1,A}(u_A)] \leq 2e^{-u_A},$$
which completes the proof.
\end{proof}

%%%%%%%%%%%%%%%%%%%%%%%%%%%%%%%%%%%%%%%%%%%%%%%%%%%%%%%%%%%%%%%%%%%%%%%%%%%%%%%%%%%%%%%%%%%%%%%%

\bibliographystyle{ba}

%%%%%%%%%%%%%%%%%%%%%%%%%%%%%%%%%%%%%%%%%%%%%%%%%%%%%%%%%%%%%%%%%%%%%%%%%%%%%%%%%%%%%%%%%%%%%%%%%%%%%
%%%%%%%%%%%%%%%%%%%%%%%%%%%%%%%%%%%%%%%%%%%%%%%%%%%%%%%%%%%%%%%%%%%%%%%%%%%%%%%%%%%%%%%%%%%%%%%%%%%%%%

%%%%%%%%%%%%%%%%%%%%%%%%

\end{document}